\documentclass{article}
\usepackage[dvipsnames]{xcolor}
\usepackage{gastex}
\usepackage{amsmath}
\usepackage{amssymb}
\usepackage{multirow}
\usepackage{multicol}
\usepackage{makeidx}
\makeindex
\usepackage{theorem}
\newtheorem{theorem}{Theorem}
\newtheorem{lemma}[theorem]{Lemma}
\newtheorem{proposition}[theorem]{Proposition}
\newtheorem{corollary}[theorem]{Corollary}
{\theorembodyfont{\rmfamily}%
  \newtheorem{example}[theorem]{Example}
   }
\newenvironment{proof}{\noindent\textit{Proof.}}
{\QED\vskip\theorempostskipamount} 
\newenvironment{proofof}[1]{\noindent\textit{Proof
    \protect{#1}.}}
                       {\QED\vskip\theorempostskipamount}
\def\petitcarre{\vrule height4pt width 4pt depth0pt}
\def\QED{\relax\ifmmode\eqno{\hbox{\petitcarre}}\else{%
  \unskip\nobreak\hfil\penalty50\hskip2em\hbox{}\nobreak\hfil
  \petitcarre
  \parfillskip=0pt \finalhyphendemerits=0\par\smallskip}
  \fi}
\newcommand\D{\mathcal{D}}
\def\F{\mathcal{F}}

\newcommand\RR{\mathcal{R}}

\newcommand{\Z}{\mathbb{Z}}
\newcommand{\R}{\mathbb{R}}
\newcommand{\Q}{\mathbb{Q}}

\let\edge\xrightarrow
\def\un(#1){\underline{#1}\,}
\DeclareMathOperator{\Card}{Card}

\DeclareMathOperator{\Div}{Div}
\DeclareMathOperator{\Sep}{Sep}

\definecolor{ivoire}{rgb}{0.99,0.99,0.8}
\usepackage{calc}
\newcounter{hours}\newcounter{minutes}
\newcommand\computetime{\setcounter{hours}{\time/60}%
  \setcounter{minutes}{\time-\value{hours}*60}%
  \thehours\,h\,\theminutes}
\newcommand\dateandtime{\today\quad\computetime}
\usepackage[hypertex,hyperindex,%pagebackref,
final]{hyperref}
\numberwithin{theorem}{section}
\numberwithin{equation}{section}
\numberwithin{figure}{section}
\numberwithin{table}{section}
\title{Interval exchanges, admissibility and branching Rauzy induction}
\author{Francesco Dolce$^1$, Dominique Perrin$^1$\\\\
$^1$Universit\'e Paris Est, LIGM}
\date{\dateandtime}
\begin{document}
%========================
%========== Lists ==================================
\makeatletter
\def\@listI{%
  \leftmargin\leftmargini
  \setlength{\parsep}{0pt plus 1pt minus 1pt}
  \setlength{\topsep}{2pt plus 1pt minus 1pt}
  \setlength{\itemsep}{0pt}
}
\let\@listi\@listI
\@listi
\def\@listii {%
  \leftmargin\leftmarginii
  \labelwidth\leftmarginii
  \advance\labelwidth-\labelsep
  \setlength{\topsep}{0pt plus 1pt minus 1pt}
}
\def\@listiii{%
  \leftmargin\leftmarginiii
  \labelwidth\leftmarginiii
  \advance\labelwidth-\labelsep
  \setlength{\topsep}{0pt plus 1pt minus 1pt}
%              \topsep    0\p@ \@plus\p@\@minus\p@
  \setlength{\parsep}{0pt} 
  \setlength{\partopsep}{1pt plus 0pt minus 1pt}
}
\makeatother
%====================
\maketitle
%========================

\begin{abstract}
We introduce a definition of admissibility for subintervals in interval exchange transformations.
Using this notion, we prove a property of the natural codings of interval exchange transformations, namely that any derived set of a regular interval exchange set is a regular interval exchange set with the same number of intervals.
Derivation is taken here with respect to return words.
We characterize the admissible intervals using a branching version of the Rauzy induction.
We also study the case of regular interval exchange transformations defined over a quadratic field and show that the set of factors of such a transformation is primitive morphic.
The proof uses an extension of a result of Boshernitzan and Carroll.
\end{abstract}

\tableofcontents

\section{Introduction}
Interval exchange transformations were introduced by Oseledec~\cite{Oseledec1966} following an earlier idea of Arnold~\cite{Arnold1963}.
Interval exchange transformations have been generalized to transformations called linear involutions by Danthony and Nogueira in~\cite{DanthonyNogueira1990} (for other generalizations, see~\cite{Skripchenko2012}).

The natural coding of interval exchange produces sequences of linear complexity, including Sturmian sequences, and this has been widely studied (see, for example~\cite{FerencziZamboni2008} or~\cite{balazimasakovapelantova2008} for small alphabets).

Rauzy has introduced in~\cite{Rauzy1979} a transformation, now called Rauzy induction (or Rauzy-Veech induction), which operates on interval exchange transformations.
It actually transforms an interval exchange transformation into another one operating on a smaller interval.
Its iteration can be viewed as a generalization of the continued fraction expansion.
The induction consists in taking the first return map of the  transformation with respect to a subinterval of the interval on which the exchange is defined.
The induced map of an interval exchange on $s$ intervals is still an interval exchange with at most $s+2$ intervals.

Rauzy introduced in~\cite{Rauzy1979} the definition of right-admissibility for an interval and characterized the right-admissible intervals as those which can be reached by the Rauzy induction.

Interval exchange transformations defined over quadratic fields have been studied by Boshernitzan and Carroll (\cite{Boshernitzan1988} and \cite{BoshernitzanCarroll1997}).
Under this hypothesis, they showed that, using iteratively the first return map on one of the intervals exchanged by the transformation, one obtains only a finite number of different new transformations up to rescaling, extending the classical Lagrange's theorem that quadratic irrationals have a periodic continued fraction expansion.

In this paper, we generize both the notion of admissible intervals and of Rauzy induction to a two-sided version.

Our main result is a characterization of the admissible intervals (Theorem~\ref{theo:birauzy2}).
We show that, in particular, intervals associated with factors of the natural coding of an interval exchange transformation are admissible (Proposition~\ref{pro:jadm}).

Our motivation is the study of the natural coding of these transformations by words, in the spirit of the research initiated in~\cite{BerstelDeFelicePerrinReutenauerRindone2012} and containing a series of other papers (for instance \cite{BertheDeFeliceLeroyPerrinReutenauerRindone2014} and \cite{BertheDeFeliceLeroyPerrinReutenauerRindone2015}).

We prove a property of the natural codings of regular interval exchange transformations (Theorem~\ref{theo:returns}) saying that the family of these sets of words is closed by derivation, an operation consisting in taking the first return words to a given word as a new alphabet.

We pay special attention to the case of interval exchange transformations defined over a quadratic field.
We prove that the family of transformations obtained from a regular interval exchange transformation by two-sided Rauzy induction is finite up to rescaling.
Moreover, we show that the related interval exchange set is obtained as the set of factors of a primitive morphic word.

The paper is organized as follows.

In Section~\ref{sec:prel} we recall some basic definitions concerning words and sets.
Return words and first return words are also introduced.

In Section~\ref{sec:ie}, we give some notions concerning interval exchange transformations.
We recall the result of Keane~\cite{Keane1975} which proves that regularity is a sufficient condition for the minimality of such a transformation (Theorem~\ref{theo:keane}).
We also introduce the natural codings of interval exchange transformations.
We define the derivate of an interval exchange set with respect to a coding morphism and we show a closure property of these sets.

In Section~\ref{sec:rauzy}, we first recall the notion of \emph{Rauzy induction} introduced in~\cite{Rauzy1979}.
We introduce a branching version of Rauzy induction. 
We prove the generalization of Rauzy's theorems to the two-sided case (Theorems~\ref{theo:birauzy1} and~\ref{theo:birauzy2}).

In Section~\ref{sec:quadratic} we generalize the result of Boshernitzan and Carroll~\cite{BoshernitzanCarroll1997}, enlarging the family of transformations obtained using induction on every admissible semi-interval.
This contains the results of \cite{BoshernitzanCarroll1997} because every semi-interval exchanged by a transformation is admissible, while for $n>2$ there are admissible semi-intervals that we cannot obtain using the induction only on the exchanged ones.
We conclude the Section showing that regular quadratic interval exchange sets are primitive morphic (Theorem~\ref{theo:morphic}).

\paragraph{Acknowledgement} This work was supported by grants from R\'egion \^Ile-de-France and ANR project Eqinocs.

\section{Preliminaries}
\label{sec:prel}

In this section, we first recall some definitions concerning words.
We recall the definition of recurrent and uniformly recurrent sets of words (see~\cite{Lothaire2002} for a more detailed presentation).
We introduce the notion of fist return words and derived words.
Derived words have been widely studied, in particular in the context of substitutive dynamics (see~\cite{Durand1998} for example) and are intimately connected with induction.

\subsection{Words and recurrent sets}
Let $A$ be a finite nonempty alphabet. All words considered below, unless stated explicitly, are supposed to be on the alphabet $A$.
We denote by $A^*$ the set of all words on $A$.
We denote by $1$  or $\varepsilon$ the empty word.
We denote by $|w|$ the length of a word $w$.
A set of words is said to be \emph{factorial} if it contains the factors of its elements.

A \emph{morphism} $f : A^*\rightarrow B^*$ is a monoid morphism from $A^*$ into $B^*$.
If $a \in A$ is such that the word $f(a)$ begins with $a$ and if $|f^n(a)|$ tends to infinity with $n$, there is a unique infinite word denoted $f^\omega(a)$ which has all words $f^n(a)$ as prefixes.
It is called a \emph{fixed point} of the morphism $f$.

A morphism $f : A^*\rightarrow A^*$ is called \emph{primitive} if there is an integer $k$ such that for all $a, b \in A$, the letter $b$ appears in $f^k(a)$.
If $f$ is a primitive morphism, the set of factors of any fixed point of $f$ is uniformly recurrent (see~\cite{PytheasFogg2002} Proposition 1.2.3 for example).

An infinite word $y$ over an alphabet $B$ is called \emph{morphic} if there exists a morphism $f$ on an alphabet $A$, a fixed point $x = f^\omega(a)$ of $f$ and a morphism
$\sigma : A^* \to B^*$ such that $y = \sigma(x)$.
If $A = B$ and $\sigma$ is the identity map, we call $y$ \emph{purely morphic}.
If $f$ is primitive we say that the word is \emph{primitive morphic}.

A  factorial set of words $F \neq \{\varepsilon\}$ is \emph{recurrent} if  for every $u, w \in F$ there is a $v \in F$ such that $uvw \in F$.
For an infinite word $x$, we denote $F(x)$ the set of factors of $x$.
An infinite word $x$ is \emph{recurrent} if for any $u \in F(x)$ there is a $v \in F(x)$ such that $uvu \in F(x)$.
As well known, for any recurrent set $F$ there is a recurrent infinite word $x$ such that $F = F(x)$ and conversely, for any recurrent infinite word $x$, the set $F(x)$ is recurrent.

Extending the definition, we say that a set $F(x)$ is \emph{morphic} (resp. \emph{purely morphic}, \emph{primitive morphic}) if the infinite word $x$ is morphic (resp. purely morphic, primitive morphic).

A recurrent set of words $F$ is said to be \emph{uniformly recurrent} if, for any word $u \in F$, there exists an integer $n \geq 1$ such that $u$ is a factor of every word of $F$ of length $n$.

\subsection{Return words and derived sets}
Let $F$ be a recurrent set.
For $w \in F$, let
$$
\Gamma_F(w) = \{x \in F \mid wx \in F \cap A^+w\}
\quad \mbox{ and } \quad
\Gamma_F(w)' = \{x \in F \mid xw \in F \cap wA^+\}
$$
be respectively the set of \emph{right return words} and \emph{left return words} to $w$.
Since $F$ is recurrent, the sets $\Gamma_S(w)$ and $\Gamma_S(w)'$ are nonempty.
Let
$$
\RR_F(w) = \Gamma_F(w) \setminus \Gamma_F(w) A^+
\quad \mbox{ and } \quad
\RR_F(w)' = \Gamma_F(w)' \setminus A^+ \Gamma_F(w)'
$$
be respectively the set of \emph{first right return words} and the set of \emph{first left return words} to $w$.
Note that $w\RR_F(w) = \RR_F(w)'w$.

Clearly, a recurrent set $F$ is uniformly recurrent if and only if the set $\RR_F(w)$ (resp. $\RR_F(w)'$) is finite for any $w \in F$.

\begin{example}
\label{ex:returns}
Let $F$ be the set defined in Example~\ref{ex:t} and whose factors of length at most 6 are represented in Figure~\ref{fig:setf}.
We have
\begin{eqnarray*}
\RR_F(a)	& = &	\{cbba,ccba,ccbba\}, \\
\RR_F(b)	& = &	\{acb,accb,b\}, \\
\RR_F(c)	& = &	\{bac,bbac,c\}.
\end{eqnarray*}
\end{example}

For a set of words $X$ and a word $u$, we denote $u^{-1}X = \{ x \in A^* \mid uv \in X \}$.
Let $F$ be a  recurrent set and let $w \in F$.
A \emph{coding morphism} for the set $\RR_F(w)$ is a morphism $f : B^*\rightarrow A^*$ which maps bijectively the (possibly infinite) alphabet $B$ onto $\RR_F(w)$.
The set $f^{-1}(w^{-1}F)$, denoted $\D_f(F)$, is called the \emph{derived set} of $F$ with respect to $f$.
The following result is proved in~\cite[Proposition 4.3]{BertheDeFeliceLeroyPerrinReutenauerRindone2015}.
\begin{proposition}
\label{pro:recurrent}
Let $F$ be a recurrent set.
For $w \in F$, let $f$ be a coding morphism for the set $\RR_F(w)$.
Then 
$$
\D_f(F) = f^{-1}(\Gamma_F(w)) \cup \{1\}.
$$
\end{proposition}

Let $F$ be a recurrent set and $x$ be an infinite word such that $F = F(x)$.
Let $w \in F$ and let $f$ be a coding morphism for the set $\RR_F(w)$.
Since $w$ appears infinitely often in $x$, there is a unique factorization $x = vwy$ with $y \in \RR_F(w)^\omega$ and $v$ such that $vw$ has no proper prefix ending with $w$.
The infinite word $f^{-1}(y)$ is called the \emph{derived word} of $x$ relative to $f$, denoted $\D_f(x)$.

Since the set of factors of a recurrent infinite word is recurrent, the following result, proved in~\cite[Proposition 4.4]{BertheDeFeliceLeroyPerrinReutenauerRindone2015}, shows in particular that the derived set of a recurrent set is recurrent.

\begin{proposition}
\label{pro:derived}
Let $F$ be a recurrent set and let $x$ be an infinite word such that $F = F(x)$.
Let $w \in F$ and let $f$ be a coding morphism for the set $\RR_F(w)$.
The derived set of $F$ with respect to $f$ is the set of factors of the derived word of $x$ with respect to $f$, that is $\D_f(F) = F(\D_f(x))$.
\end{proposition}

\begin{example}
Let $F$ be the uniformly recurrent set of Example~\ref{ex:t} and whose factors of length at most 6 are represented in Figure~\ref{fig:setf}.
Let $f$ be the coding morphism for the set $\RR_F(c)$ given by $f(a) = bac$, $f(b) = bbac$, $f(c) = c$.
The derived set of $F$ with respect to $f$ is represented in Figure~\ref{fig:derived}.

\begin{figure}[hbt]
\centering
\gasset{AHnb=0,Nadjust=wh}
\begin{picture}(50,35)(0,5)
\node(1)(0,20){}

\node(a)(10,30){}
\node(b)(10,20){}
\node(c)(10,10){}

\node(ac)(20,30){}
\node(bb)(20,25){}
\node(bc)(20,15){}
\node(ca)(20,10){}
\node(cb)(20,5){}

\node(aca)(30,35){}
\node(acb)(30,30){}
\node(bbb)(30,25){}
\node(bbc)(30,20){}
\node(bcb)(30,15){}
\node(cac)(30,10){}
\node(cbb)(30,5){}

\drawedge(1,a){$a$}
\drawedge[ELpos=60](1,b){$b$}
\drawedge[ELpos=60,ELside=r](1,c){$c$}
\drawedge(a,ac){$c$}
\drawedge(b,bb){$b$}
\drawedge[ELpos=60](b,bc){$\ c$}
\drawedge(c,ca){$a$}
\drawedge[ELpos=60,ELside=r](c,cb){$b$}
\drawedge(ac,aca){$a$}
\drawedge[ELpos=60](ac,acb){$\ \ \ b$}
\drawedge(bb,bbb){$b$}
\drawedge[ELside=r](bb,bbc){$c$}
\drawedge(bc,bcb){$b$}
\drawedge(ca,cac){$c$}
\drawedge[ELside=r](cb,cbb){$b$}
\end{picture}
\caption{The words of length $ \leq 3$ of $\D_f(F)$.}
\label{fig:derived}
\end{figure}
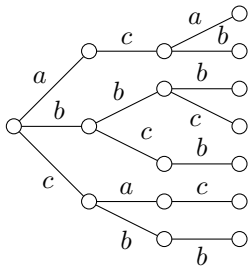
\end{example}

\section{Interval exchanges}
\label{sec:ie}
In this section we recall the basic definitions of interval exchange transformations, along with minimality and regularity of interval exchanges.
We also define the natural coding of an interval exchange and the associated interval exchange set, that is the language of all possible natural codings of a transformation.
We then consider the notion of admissibility of a semi-interval.
Finally we prove a closure property using the derivation defined in the previous section.

\subsection{Interval exchange transformations}
Let us recall the definition of an interval exchange transformation (see~\cite{CornfeldFominSinai1982}, \cite{Yoccoz2010} or \cite{Viana2006} for a more detailed presentation).

Let $A$ be a finite, nonempty and ordered alphabet.
All words considered below, unless stated explicitly, are supposed to be on the alphabet $A$.
We denote by $A^*$ the set of all words on $A$ and by $1$ or $\varepsilon$ the empty word.

A \emph{semi-interval} is a nonempty subset of the real line of the form $[\alpha, \beta[ = \{ z \in \R \mid \alpha \leq z < \beta \}$.
Thus it is a left-closed and right-open interval.
For two semi-intervals $\Delta, \Gamma$, we denote $\Delta < \Gamma$ if $x < y$ for any $x\in\Delta$ and $y\in\Gamma$. 

Given an order $<$ on $A$, a partition $(I_a)_{a \in A}$ of a semi-interval $[\ell,r[$ in semi-intervals is \emph{ordered} if $a < b$ implies $I_a < I_b$.

Let now $<_1$ and $<_2$ be two total orders on $A$.
Let $(I_a)_{a \in A}$ be  a partition of $[\ell, r[$ in semi-intervals ordered for $<_1$.
Let $\lambda_a$ be the length of $I_a$.
Let $\mu_a = \sum_{b \leq_1 a} \lambda_b$ and $\nu_a = \sum_{b \leq_2 a} \lambda_b$.
Set $\alpha_a = \nu_a - \mu_a$.
The \emph{interval exchange transformation} relative to $(I_a)_{a \in A}$ is the map $T:[\ell, r[ \rightarrow [\ell, r[$ defined by
$$
T(z) = z + \alpha_a \quad \text{ if } z \in I_a.
$$
Observe that the restriction of $T$ to $I_a$ is a translation onto $J_a = T(I_a)$, that $\mu_a$ is the right boundary of $I_a$ and that $\nu_a$ is the right boundary of $J_a$.
We additionally denote by $\gamma_a$ the left boundary of $I_a$ and by $\delta_a$ the left boundary of $J_a$.
Thus
$$
I_a = [\gamma_a, \mu_a[, \quad J_a =[\delta_a, \nu_a[.
$$

Since $a<_2 b$ implies $J_a<_2J_b$, the family $(J_a)_{a \in A}$ is a partition of $[\ell, r[$ ordered for $<_2$.
In particular, the transformation $T$ defines a bijection from $[\ell, r[$ onto itself.

An interval exchange transformation relative to $(I_a)_{a \in A}$ is also said to be on the alphabet $A$.
The values $(\alpha_a)_{a \in A}$ are called the \emph{translation values} of the transformation $T$.

\begin{example}
\label{ex:rotation}
Let $R$ be the interval exchange transformation corresponding to $A = \{a, b\}$, $ a<_1 b$, $b <_2 a$, $I_a = [0, 1-\alpha[$, $I_b = [1-\alpha, 1[$.
The transformation $R$ is the rotation of angle $\alpha$ on the semi-interval $[0,1[$ defined by $R(z) = z + \alpha \bmod 1$.
\end{example}

Since $<_1$ and $<_2$ are total orders, there exists a unique permutation $\pi$ of $A$ such that $a <_ 1b$ if and only if $\pi(a) <_2 \pi(b)$.
Conversely, $<_2$ is determined by $<_1$ and $\pi$ and $<_1$ is determined by $<_2$ and $\pi$.
The permutation $\pi$ is said to be \emph{associated} to $T$.

Set $A = \{a_1, a_2, \ldots, a_s\}$ with $a_1 <_1 a_2 <_1 \cdots <_1 a_s$.
The pair $(\lambda, \pi)$ formed by the family $\lambda = (\lambda_a)_{a \in A}$ and the permutation $\pi$ determines the map $T$.
We will also denote $T$ as $T_{\lambda, \pi}$.
The transformation $T$ is also said to be an $s$-interval exchange transformation.

It is easy to verify that the family of $s$-interval exchange transformations is closed by  taking inverses.

\begin{example}
\label{ex:3iet}
Let $T=R^2$ where $R$ is the rotation of Example~\ref{ex:rotation}.
The transformation $T$, represented in Figure~\ref{fig:3iet} is a $3$-interval exchange transformation.
One has $A = \{ a, b, c \}$ with $a <_1 b <_1 c$ and $b <_2 c <_2 a$.
The associated permutation is the cycle $\pi = (abc)$.

\begin{figure}[hbt]
\centering
\gasset{Nh=2,Nw=2,ExtNL=y,NLdist=2,AHnb=0,ELside=r}
\begin{picture}(100,15)
\node[fillcolor=red](0h)(0,10){$0$}
\node[fillcolor=blue](1-2alpha)(23.6,10){$1-2\alpha$}
\node[fillcolor=green](1-alpha)(61.8,10){$1-\alpha$}
\node(1h)(100,10){$1$}
\drawedge[linecolor=red,linewidth=1](0h,1-2alpha){$a$}
\drawedge[linecolor=blue,linewidth=1](1-2alpha,1-alpha){$b$}
\drawedge[linecolor=green,linewidth=1](1-alpha,1h){$c$}

\node[fillcolor=blue](0b)(0,0){$0$}
\node[fillcolor=green](alpha)(38.2,0){$\alpha$}
\node[fillcolor=red](2alpha)(76.4,0){$2\alpha$}\node(1b)(100,0){$1$}
\drawedge[linecolor=blue,linewidth=1](0b,alpha){$b$}
\drawedge[linecolor=green,linewidth=1](alpha,2alpha){$c$}
\drawedge[linecolor=red,linewidth=1](2alpha,1b){$a$}
\end{picture}
\caption{A $3$-interval exchange transformation.}
\label{fig:3iet}
\end{figure}
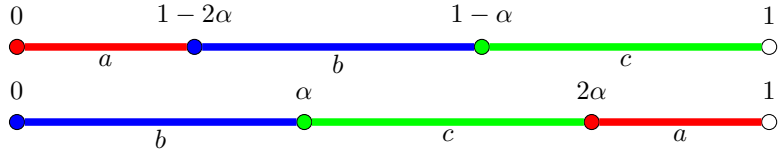
\end{example}

The \emph{orbit} of a point $z \in [\ell,r[$ is the set $\mathcal{O}(z) = \{ T^n(z) \mid n \in \mathbb{Z} \}$.
The transformation $T$ is said to be \emph{minimal} if for any $z \in [\ell,r[$, the orbit of $z$ is dense in $[\ell,r[$.

From now on, set $\gamma_i = \gamma_{a_i}$, $\delta_i = \delta_{a_i}$, $\mu_i = \mu_{a_i}$ and $\nu_i = \nu_{a_i}$.
The points $0 = \gamma_1, \mu_1 = \gamma_2, \ldots, \mu_{s-1} = \gamma_s$ form the set of \emph{separation points} of $T$, denoted $\Sep(T)$.
Note that the transformation $T$ has at most $s-1$ \emph{singularities} (that is points at which it is not continuous), which are among the nonzero separation points $\gamma_2, \ldots, \gamma_s$.

An interval exchange transformation $T_{\lambda, \pi}$ is called \emph{regular} if the orbits of the nonzero separation points $\gamma_2, \ldots,\gamma_s$ are infinite and disjoint.
Note that the orbit of $0$ cannot be disjoint from the others since one has $T(\gamma_i)=0$ for some $i$ with $2 \leq i \leq s$.

A regular interval exchange transformation is also said to satisfy the \emph{idoc} condition (where idoc stands for infinite disjoint orbit condition).
It is also said to have the Keane property or to be without \emph{connection} (see~\cite{BoissyLanneau2009}).
As an example, the $2$-interval exchange transformation of Example~\ref{ex:rotation} which is the rotation of angle $\alpha$ is regular if and only if $\alpha$ is irrational.

The following result is due to Keane~\cite{Keane1975}.

\begin{theorem}[Keane]
\label{theo:keane}
A regular interval exchange transformation is minimal.
\end{theorem}

The converse is not true.
Indeed, consider the rotation of angle $\alpha$ with $\alpha$ irrational, as a $3$-interval exchange transformation with $\lambda = (1-2\alpha, \alpha, \alpha)$ and
$\pi = (132)$.
The transformation is minimal as any rotation of irrational angle but it is not regular since $\mu_1 = 1-2\alpha$, $\mu_2 = 1-\alpha$ and thus $\mu_2 = T(\mu_1)$.

\begin{example}
\label{ex:2alpha}
Let $T$ be the $3$-interval exchange transformation of Example~\ref{ex:3iet} with $\alpha=(3-\sqrt{5})/2$.
The transformation $T$ is regular since $\alpha$ is irrational.
Note that $1-\alpha$ is a separation point which is not a singularity since $T$ is also a $2$-interval exchange transformation.
\end{example}

The following necessary condition for minimality of an interval exchange transformation is useful.
A permutation $\pi$ of an ordered set $A$ is called \emph{decomposable} if there exists an element $b \in A$ such that the set $B$ of elements strictly less than $b$ is nonempty and such that $\pi(B) = B$.
Otherwise it is called \emph{indecomposable}.
If an interval exchange transformation $T=T_{\lambda, \pi}$ is minimal, the permutation $\pi$ is indecomposable.
Indeed, if $B$ is a set as above, the set of orbits of the points in the set $S = \cup_{a \in B}I_a$ is closed and strictly included in $[\ell,r[$.
The following example shows that the indecomposability of $\pi$ is not sufficient for $T$ to be minimal.

\begin{example}
Let $A = \{a, b, c\}$ and $\lambda$ be such that $\lambda_a = \lambda_c$.
Let $\pi$ be the transposition $(ac)$.
Then $\pi$ is indecomposable but $T_{\lambda, \pi}$ is not minimal since it is the identity on $I_b$.
\end{example}

The iteration of an $s$-interval exchange transformation is, in general, an interval exchange transformation operating on a larger number of semi-interval.

\begin{proposition}
\label{pro:tn}
Let $T$ be a regular $s$-interval exchange transformation.
Then, for any $n \geq 1$, $T^n$ is a regular $n(s-1)+1$-interval exchange transformation.
\end{proposition}
\begin{proof}
Since $T$ is regular, the set $\cup_{i=0}^{n-1} T^{-i}(\mu)$ where $\mu$ runs over the set of $s-1$ nonzero separation points of $T$ has $n(s-1)$ elements.
These points partition the interval $[\ell,r[$ in $n(s-1)+1$ semi-intervals on which $T$ is a translation.
\end{proof}

We close this subsection with a lemma that will be useful in Section~\ref{sec:ie}.

\begin{lemma}
\label{lem:distance}
Let $T$ be a minimal interval exchange transformation. For every $N > 0$ there exists an $\varepsilon > 0$ such that for every $z \in D(T)$ and for every $n>0$, one has
$$ \left| T^n(z) - z \right| < \varepsilon \quad \Longrightarrow \quad n \geq N.$$
\end{lemma}
\begin{proof}
Let $\alpha_1, \alpha_2, \ldots, \alpha_s$ be the translation values of $T$.
For every $N>0$ it is sufficient to choose
$$\varepsilon = \min \left\{ \left| \textstyle{\sum_{i_j = 1}^M \alpha_{i_j}} \right|  \; \mid \;\;\; 1 \leq i_j \leq s \; \mbox{ and } \; M \leq N \right\}.$$
\end{proof}

\subsection{Natural coding}
\label{subsec:nc}
Let $T$ be an interval exchange transformation relative to $(I_a)_{a\in A}$.
For a given real number $z \in [\ell,r[$, the \emph{natural coding} of $T$ relative to $z$ is the infinite word $\Sigma_T(z) = a_0 a_1 \cdots$ on the alphabet $A$ defined by
$$
a_n = a \quad \text{ if } \quad T^n(z) \in I_{a}.
$$

\begin{example}
\label{ex:fibocoding}
Let $\alpha = (3-\sqrt{5})/2$ and let $R$ be the rotation of angle $\alpha$ on $[0,1[$ as in Example~\ref{ex:rotation}.
The natural coding of $R$ relative to $\alpha$ is the Fibonacci word which is the fixed point $t = abaab\cdots$ of the morphism $f$ from $\{a,b\}^*$ into itself defined by $f(a) = ab$ and $f(b) = a$ (see, for example,~\cite[Chapter 2]{Lothaire2002}).
\end{example}

For a word $w = b_0 b_1 \cdots b_{m-1}$, let $I_w$ be the set
\begin{equation}
I_w = I_{b_0} \cap T^{-1}(I_{b_1}) \cap \ldots \cap T^{-m+1}(I_{b_{m-1}}).
\label{eq:iu}
\end{equation}

Note that each $I_w$ is a semi-interval.
Indeed, this is true if $w$ is a letter.
Next, assume that $I_w$ is a semi-interval.
Then for any $a\in A$, $T(I_{aw}) = T(I_a) \cap I_w$ is a semi-interval since $T(I_a)$ is a semi-interval by definition of an interval exchange transformation.
Since $I_{aw} \subset I_a$, $T(I_{aw})$ is a translate of $I_{aw}$, which is therefore also a semi-interval.
This proves the property by induction on the length.
The semi-interval $I_w$ is the set of points $z$ such that the natural coding of the transformation relative to $z$ has $w$ as a prefix.

Set $J_w = T^m(I_w)$.
Thus
\begin{equation}
J_w = T^{m}(I_{b_0}) \cap T^{m-1}(I_{b_1}) \cap \ldots \cap T(I_{b_{m-1}}).
\label{eq:ju}
\end{equation}
In particular, we have $J_a = T(I_a)$ for $a \in A$.
Note that each $J_w$ is a semi-interval.
Indeed, this is true if $w$ is a letter.
Next, for any $a \in A$, we have $T^{-1}(J_{wa}) = J_w \cap I_a$.
This implies as above that $J_{wa}$ is a semi-interval and proves the property by induction.
We set by convention $I_\varepsilon = J_\varepsilon=[0,1[$.
Then one has for any $n \geq 0$
\begin{equation}
a_n a_{n+1} \cdots a_{n+m-1} =w \Longleftrightarrow T^n(z) \in I_w
\label{eq:iw}
\end{equation}
and
\begin{equation}
a_{n-m} a_{n-m+1} \cdots a_{n-1} = w \Longleftrightarrow T^n(z) \in J_w
\label{eq:jw}
\end{equation}
Let $(\alpha_a)_{a \in A}$ be the translation values of $T$.
Note that for any word $w$,
\begin{equation}
J_w = I_w + \alpha_w
\label{eq:jw2}
\end{equation}
with $\alpha_w = \sum_{j=0}^{m-1}\alpha_{b_j}$ as one may verify by induction on $|w| = m$.
Indeed it is true for $m = 1$.
For $m \geq 2$, set $w = ua$ with $a = b_{m-1}$.
One has $T^m(I_w) = T^{m-1}(I_w) + \alpha_{a}$ and $T^{m-1}(I_w) = I_w + \alpha_u$ by the induction hypothesis and the fact that $I_w$ is included in $I_u$.
Thus $J_w = T^m(I_w) = I_w + \alpha_u + \alpha_a = I_w + \alpha_w$.
Equation~\eqref{eq:jw2} shows in particular that the restriction of $T^{|w|}$ to $I_w$ is a translation.

Note that the semi-interval $J_w$ is the set of points $z$ such that the natural coding of $T^{-|w|}(z)$ has $w$ as a prefix.
%the transformation relative to $z$ has $w^{-1}$ as a prefix, in the sense that for every prefix $u$ of $\Sigma_T(z)$, one has $wu \in F$.

If $T$ is minimal, one has $w \in F(\Sigma_T(z))$ if and only if $I_w \neq \emptyset$.
Thus the set $F(\Sigma_T(z))$ does not depend on $z$ (as for Sturmian words, see~\cite{Lothaire2002}).
Since it depends only on $T$, we denote it by $F(T)$.
When $T$ is regular (resp. minimal), such a set is called a \emph{regular interval exchange set} (resp. a minimal interval exchange set).

Let $X$ be the closure of the set of all $\Sigma_T(z)$ for $z \in [\ell,r[$ and let $S$ be the shift on $X$ defined by $S(x) = y$ with $y_n = x_{n+1}$ for $n \geq 0$. The pair $(X,S)$ is a \emph{symbolic dynamical system}, formed of a topological space $X$ and a continuous transformation $S$.
Such a system is said to be minimal if the only closed subsets invariant by $S$ are $\emptyset$ or $X$.
It is well-known that $(X,S)$ is minimal if and only if $F(S)$ is uniformly recurrent (see for example~\cite{Lothaire2002} Theorem 1.5.9).

Then we have the following commutative diagram of Figure~\ref{fig:diagram}.

\begin{figure}[hbt]
\centering
\gasset{Nframe=n}
\begin{picture}(20,20)
\node(T1)(0,20){$[\ell,r[$}
\node(T2)(20,20){$[\ell,r[$}
\node(S1)(0,0){$X$}
\node(S2)(20,0){$X$}
\drawedge(T1,T2){$T$}
\drawedge(T1,S1){$\Sigma_T$}
\drawedge(S1,S2){$S$}
\drawedge(T2,S2){$\Sigma_T$}
\end{picture}
\caption{A commutative diagram.}
\label{fig:diagram}
\end{figure}

The map $\Sigma_T$ is neither continuous nor surjective.
This can be corrected by embedding the interval $[\ell,r[$ into a larger space on which $T$ is a homeomorphism (see~\cite{Keane1975} or~\cite{BertheRigo2010} page 349).
However, if the transformation $T$ is minimal, the symbolic dynamical system $(X,S)$ is minimal (see~\cite{BertheRigo2010} page 392).
Thus, we obtain the following statement.

\begin{proposition}
\label{pro:regularur}
For any minimal interval exchange transformation $T$, the set $F(T)$ is uniformly recurrent.
\end{proposition}

Note that for a regular interval exchange transformation $T$, the map $\Sigma_T$ is injective (see \cite{Keane1975} page 30).

\begin{example}
\label{ex:t}
Let $T$ be the transformation of Example~\ref{ex:2alpha}.
Since $T$ is minimal, the set $F(T)$ is uniformly recurrent.
In Subsection~\ref{subsec:morphic} we will show that the set $F = F(T)$ is the set of factors of the fixed point of a primitive morphism. 
The words of length at most $6$ of the set $F$ are represented in Figure~\ref{fig:setf}.

\begin{figure}[hbt] 
\centering
\gasset{Nadjust=wh,AHnb=0,ELpos=60,ELdist=.6}
\begin{picture}(60,50)(0,5)
\node(1)(0,30){}
\node(a)(10,40){}
\node(b)(10,30){}
\node(c)(10,20){}

\node(ac)(20,45){}
\node(ba)(20,35){}
\node(bb)(20,25){}
\node(cb)(20,20){}
\node(cc)(20,10){}

\node(acb)(30,50){}
\node(acc)(30,45){}
\node(bac)(30,35){}
\node(bba)(30,25){}
\node(cba)(30,20){}
\node(cbb)(30,15){}
\node(ccb)(30,10){}

\node(acbb)(40,55){}
\node(accb)(40,45){}
\node(bacb)(40,40){}
\node(bacc)(40,35){}
\node(bbac)(40,25){}
\node(cbac)(40,20){}
\node(cbba)(40,15){}
\node(ccba)(40,10){}
\node(ccbb)(40,5){}

\node(acbba)(50,55){}
\node(accba)(50,50){}
\node(accbb)(50,45){}
\node(bacbb)(50,40){}
\node(baccb)(50,35){}
\node(bbacb)(50,30){}
\node(bbacc)(50,25){}
\node(cbacc)(50,20){}
\node(cbbac)(50,15){}
\node(ccbac)(50,10){}
\node(ccbba)(50,5){}

\node(acbbac)(60,55){}
\node(accbac)(60,50){}
\node(accbba)(60,45){}
\node(bacbba)(60,40){}
\node(baccba)(60,37){}
\node(baccbb)(60,32){}
\node(bbacbb)(60,30){}
\node(bbaccb)(60,25){}
\node(cbaccb)(60,20){}
\node(cbbacb)(60,17){}
\node(cbbacc)(60,13){}
\node(ccbacc)(60,10){}
\node(ccbbac)(60,5){}

\drawedge(1,a){$a$}
\drawedge(1,b){$b$}
\drawedge(1,c){$c$}
\drawedge(a,ac){$c$}
\drawedge(b,ba){$a$}
\drawedge(b,bb){$b$}
\drawedge(c,cb){$b$}
\drawedge(c,cc){$c$}
\drawedge(ac,acb){$b$}
\drawedge(ac,acc){$c$}
\drawedge(ba,bac){$c$}
\drawedge(bb,bba){$a$}
\drawedge(cb,cba){$a$}
\drawedge(cb,cbb){$b$}
\drawedge(cc,ccb){$b$}
\drawedge(acb,acbb){$b$}
\drawedge(acc,accb){$b$}
\drawedge(bac,bacb){$b$}
\drawedge(bac,bacc){$c$}
\drawedge(bba,bbac){$c$}
\drawedge(cba,cbac){$c$}
\drawedge(cbb,cbba){$a$}
\drawedge(cbba,cbbac){$c$}
\drawedge(ccb,ccba){$a$}
\drawedge(ccb,ccbb){$b$}
\drawedge(acbb,acbba){$a$}
\drawedge(accb,accba){$a$}
\drawedge(accb,accbb){$b$}
\drawedge(bacb,bacbb){$b$}
\drawedge(bacc,baccb){$b$}
\drawedge(bbac,bbacb){$b$}
\drawedge(bbac,bbacc){$c$}
\drawedge(cbac,cbacc){$c$}
\drawedge(cbba,cbbac){$c$}
\drawedge(ccba,ccbac){$c$}
\drawedge(ccbb,ccbba){$a$}
\drawedge(acbba,acbbac){$c$}
\drawedge(accba,accbac){$c$}
\drawedge(accbb,accbba){$a$}
\drawedge(bacbb,bacbba){$a$}
\drawedge(baccb,baccba){$a$}
\drawedge(baccb,baccbb){$b$}
\drawedge(bbacb,bbacbb){$b$}
\drawedge(bbacc,bbaccb){$b$}
\drawedge(cbacc,cbaccb){$b$}
\drawedge(cbbac,cbbacb){$b$}
\drawedge(cbbac,cbbacc){$c$}
\drawedge(ccbac,ccbacc){$c$}
\drawedge(ccbba,ccbbac){$c$}
\end{picture}
\caption{The words of length $ \leq 6$ of the set $F$.}
\label{fig:setf}
\end{figure}
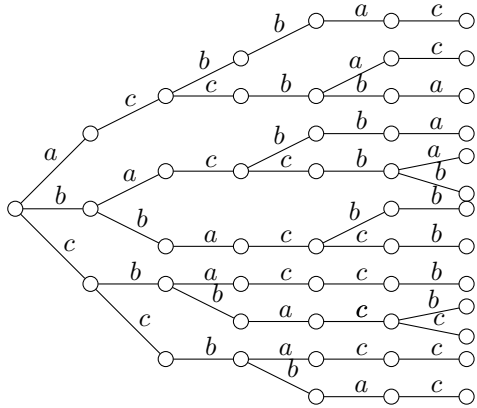
\end{example}

In Section~\ref{sec:quadratic} we will give a sufficient condition for an interval exchange set to be primitive morphic (Theorem~\ref{theo:morphic}).

\subsection{Induced transformations and admissible semi-intervals}
Let $T$ be a minimal interval exchange transformation.
Let $I\subset [\ell, r[$ be a  semi-interval.
Since $T$ is minimal, for each $z \in [\ell, r[$ there is an integer $n>0$ such that $T^n(z) \in I$.

The \emph{transformation induced} by $T$ on $I$ is the transformation $S : I\rightarrow I$ defined for $z \in I$ by $S(z) = T^n(z)$ with $n = \min\{ n>0 \mid T^n(z) \in I\}$. We also say that $S$ is the \emph{first return map} (of $T$) on $I$.
The semi-interval $I$ is called the \emph{domain} of $S$, denoted $D(S)$.

\begin{example}
Let $T$ be the transformation of Example~\ref{ex:2alpha}.
Let $I = [0, 2\alpha[$.
The transformation induced by $T$ on $I$ is
$$
S(z) =
\begin{cases}
T^2(z)	&	\text{ if $0 \leq z< 1-2\alpha$} \\
T(z)	&	\text{ otherwise}.
\end{cases}
$$
\end{example}

Let $T = T_{\lambda, \pi}$ be an interval exchange transformation relative to  $(I_a)_{a \in A}$.
For  $\ell < t < r$, the semi-interval $[\ell, t[$  is \emph{right admissible} for $T$ if there is a $k \in \Z$ such that $t = T^k(\gamma_a)$ for some $a\in A$ and
\begin{enumerate}
\item[(i)] if $k > 0$, then $t < T^h(\gamma_a)$ for all $h$ such that $0 < h < k$,
\item[(ii)] if $k \leq 0$, then $t < T^h(\gamma_a)$ for all $h$ such that $k < h \leq 0$.
\end{enumerate}
We also say that $t$ itself is right admissible.
Note that all semi-intervals $[\ell, \gamma_a[$ with $\ell < \gamma_a$ are right admissible.
Similarly, all semi-intervals $[\ell, \delta_a[$ with $\ell < \delta_a$ are right admissible.

\begin{example}
\label{ex:division}
Let $T$ be the interval exchange transformation of Example~\ref{ex:2alpha}.
The semi-interval $[0, t[$ for $t = 1-2\alpha$ or $t = 1-\alpha$ is right admissible since $1-2\alpha = \gamma_b$ and $1-\alpha = \gamma_c$.
On the contrary, for $t = 2-3\alpha$, it is not right admissible because $t = T^{-1}(\gamma_c)$ but $\gamma_c < t$ contradicting (ii).
\end{example}

The following result is Theorem 14 in~\cite{Rauzy1979}.

\begin{theorem}[Rauzy]
\label{theo:rauzy1}
Let $T$ be a regular $s$-interval exchange transformation and let $I$ be a right admissible interval for $T$.
The transformation induced by $T$ on $I$ is a regular $s$-interval exchange
transformation.
\end{theorem}

Note that the transformation induced by an $s$-interval exchange transformation on $[\ell, r[$ on any semi-interval included in $[\ell, r[$ is always an interval exchange transformation on at most $s+2$ intervals (see~\cite{CornfeldFominSinai1982}, Chapter 5 p. 128).

\begin{example}
\label{ex:induced1}
Consider again the transformation of Example~\ref{ex:2alpha}.
The transformation induced by $T$ on the semi-interval $I = [0,2\alpha[$ is the $3$-interval exchange transformation represented in Figure~\ref{fig:3ietind}.
\begin{figure}[hbt]
\centering
\gasset{Nw=2,Nh=2,ExtNL=y,NLdist=2,AHnb=0,ELside=r}
\begin{picture}(80,15)
\node[fillcolor=red](0h)(0,10){$0$}
\node[fillcolor=blue](1-2alpha)(23.6,10){$1-2\alpha$}
\node[fillcolor=green](1-alpha)(61.8,10){$1-\alpha$}
\node(2alphah)(76.4,10){$2\alpha$}
\drawedge[linecolor=red,linewidth=1](0h,1-2alpha){$a$}
\drawedge[linecolor=blue,linewidth=1](1-2alpha,1-alpha){$b$}
\drawedge[linecolor=green,linewidth=1](1-alpha,2alphah){$c$}

\node[fillcolor=blue](0b)(0,0){$0$}
\node[fillcolor=green](alpha)(38.2,0){$\alpha$}
\node[fillcolor=red](4alpha-1)(52.8,0){$4\alpha-1$}
\node(2alpha)(76.4,0){$2\alpha$}
\drawedge[linecolor=blue,linewidth=1](0b,alpha){$b$}
\drawedge[linecolor=green,linewidth=1](alpha,4alpha-1){$c$}
\drawedge[linecolor=red,linewidth=1](4alpha-1,2alpha){$a$}
\end{picture}
\caption{The transformation induced on $I$.}
\label{fig:3ietind}
\end{figure}
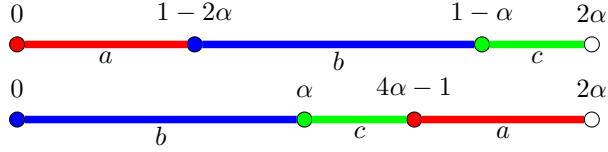
\end{example}

The notion of left admissible interval is symmetrical to that of right admissible.
For $\ell < t < r$, the semi-interval $[t, r[$ is \emph{left admissible} for $T$ if there is a $k \in \Z$ such that $t = T^k(\gamma_a)$ for some $a \in A$ and
\begin{enumerate}
\item[(i)] if $k > 0$, then $T^h(\gamma_a) < t$ for all $h$ such that $0 < h < k$,
\item[(ii)] if $k \leq 0$, then $T^h(\gamma_a) < t$ for all $h$ such that $k <h \leq 0$.
\end{enumerate}
We also say that $t$ itself is left admissible.
Note that, as for right induction, the semi-intervals $[\gamma_a, r[$ and $[\nu_a, r[$ are left admissible.
The symmetrical statements of Theorem~\ref{theo:rauzy1} also hold for left admissible intervals.

Let now generalize the notion of admissibility to a two-sided version.
For a semi-interval $I = [u,v[ \ \subset [\ell, r[$, we define the following functions on $[\ell, r[$:
$$
\rho^+_{I,T}(z) = \min\{ n > 0 \mid T^n(z) \in\ ]u,v[ \}, \quad \rho^-_{I,T}(z) = \min\{ n \geq 0 \mid T^{-n}(z) \in\ ]u,v[\}.
$$

We then define three sets.
First, let
$$
E_{I,T}(z) = \{k \mid -\rho_{I,T}^-(z) \leq k < \rho_{I,T}^+(z)\}.
$$
Next, the set of \emph{neighbors} of $z$ with respect to $I$ and $T$ is
$$
N_{I,T}(z) = \{T^k(z) \mid k \in E_{I,T}(z)\}.
$$
The set of \emph{division points} of $I$ with respect to $T$ is the finite set
$$
\Div(I,T) = \bigcup_{i=1}^s N_{I,T}(\gamma_i).
$$

We now formulate the following definition.
For $\ell \leq u < v \leq r$, we say that the semi-interval $I = [u,v[$ is \emph{admissible} for $T$ if $u,v \in \Div(I,T) \cup \{r\}$.

Note that a semi-interval $[\ell, v[$ is right admissible if and only if  it is admissible and that a semi-interval $[u, r[$ is left admissible if and only if it is
admissible.
Note also that $[\ell, r[$ is admissible.

Note also that for a regular interval exchange transformation relative to a partition $(I_a)_{a \in A}$, each of the semi-intervals $I_a$ (or $J_a$) is admissible although only the first one is right admissible (and the last one is left admissible).
Actually, we can prove that for every word $w$, the semi-intervals $I_w$ and $J_w$ are admissible.
In order to do that, we need the following Lemma.

\begin{lemma}
\label{lem:boundary}
Let $T$ be a $s$-interval exchange transformation on the semi-interval $[\ell,r[$.
For any $k \geq 1$, the set  $P_{k} = \{T^h(\gamma_i) \mid 1 \leq i \leq s,\ 1 \leq h \leq k\}$ is the set of $(s-1)k+1$ left boundaries of the semi-intervals $J_y$ for
all words $y$ with $|y| = k$.
\end{lemma}
\begin{proof}
Let $Q_k$ be the set of left boundaries of the intervals $J_y$ for $|y| = k$.
Since $\Card(F(T) \cap A^k) = (s-1)k+1$ by Proposition~\ref{pro:tn}, we have $\Card(Q_k) = (s-1)k+1$.
Since $T$ is regular the set $R_k = \{T^h(\gamma_i) \mid 2 \leq i \leq s,\ 1 \leq h \leq k\}$ is made of $(s-1)k$ distinct points.
Moreover, since 
$$
\gamma_1 = T(\gamma_{\pi(1)}),\ T(\gamma_1) = T^2(\gamma_{\pi(1)}), \ldots, T^{k-1}(\gamma_1) = T^k(\gamma_{\pi(1)}),
$$
we have $P_k = R_k\cup \{T^k(\gamma_1)\}$.
This implies $\Card(P_k) \leq (s-1)k+1$.
On the other hand, if $y = b_0 \cdots b_{k-1}$, then $J_y = \cap_{i=0}^{k-1} T^{k-i}(I_{b_{i}})$.
Thus the left boundary of each $J_y$ is the left boundary of some $T^h(I_a)$ for some $h$ with $1 \leq h \leq k$ and some $a \in A$.
Consequently $Q_k \subset P_k$.
This proves that $\Card(P_k) = (s-1)k+1$ and that consequently $P_k = Q_k$.
\end{proof}

A dual statement holds for the semi-intervals $I_y$.

\begin{proposition}
\label{pro:jadm}
Let $T$ be a $s$-interval exchange transformation on the semi-interval $[\ell,r[$.
For any $w\in F(T)$, the semi-interval $J_w$ is admissible.
\end{proposition}
\begin{proof}
Set $|w| = k$ and $J_w = [u,v[$.
By Lemma~\ref{lem:boundary}, we have $u = T^g(\gamma_i)$ for $1 \leq i \leq s$ and $1 \leq g \leq k$.
Similarly, we have $v = r$ or $v = T^d(\gamma_j)$ for $1 \leq j \leq s$ and $1 \leq d \leq k$.

For  $1 < h < g$, the point $T^h(\gamma_i)$ is the left boundary of some semi-interval $J_y$ with $|y| = k$ and thus $T^h(\gamma_i) \notin J_w$.
This shows that $g \in E_{J_w,T}(\gamma_i)$ and thus that $u \in \Div(J_w,T)$.

If $v = r$, then $v \in \Div(J_w,T)$.
Otherwise, one shows in the same way as above that $v \in \Div(J_w,T)$.
Thus $J_w$ is admissible.
\end{proof}
Note that the same statement holds for the semi-intervals $I_w$ instead of the semi-intervals $J_w$ (using the dual statement of Lemma~\ref{lem:boundary}).

It can be useful to reformulate the definition of a division point and of an admissible pair using the terminology of graphs.
Let $G(T)$ be the graph with vertex set $[\ell, r[$ and edges the pairs $(z, T(z))$ for $z \in [\ell, r[$.
Then, if $T$ is minimal and $I$ is a semi-interval, for any $z \in [\ell, r[$, there is a path $P_{I,T}(z)$ such that its origin $x$ and its end $y$ are in $I$, $z$ is on the path, $z \neq y$ and no vertex of the path except $x, y$ are in $I$ (actually $x = T^{-n}(z)$ with $n = \rho^-_{I,T}(z)$ and $y = T^m(z)$ with $m = \rho^+_{I,T}(z)$).
Then the division points of $I$ are the vertices which are on a path $P_{I,T}(\gamma_i)$ but not at its end (see Figure~\ref{fig:adm}).

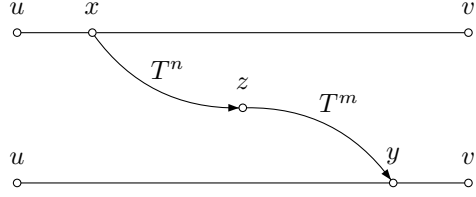
\begin{figure}[hbt]
\centering
\gasset{Nw=1,Nh=1,ExtNL=y,NLdist=2}
\begin{picture}(80,25)
\node(u1)(0,20){$u$}\node(x)(10,20){$x$}\node(v1)(60,20){$v$}
\node(z)(30,10){$z$}
\node(u3)(0,0){$u$}\node(y)(50,0){$y$}\node(v3)(60,0){$v$}
\node[Nframe=n](d1)(80,20){}\node[Nframe=n](d2)(80,10){}\node[Nframe=n](d3)(80,0){}

\drawedge[AHnb=0](u1,x){}\drawedge[AHnb=0](x,v1){}
\drawedge[AHnb=0](u3,y){}\drawedge[AHnb=0](y,v3){}
\drawedge[curvedepth=-3](x,z){$T^n$}\drawedge[curvedepth=3](z,y){$T^m$}
\end{picture}
\caption{The neighbors of $z$ with respect to $I = [u,v[$.}
\label{fig:adm}
\end{figure}

The following is a generalization of Theorem~\ref{theo:rauzy1}.
Recall that $\Sep(T)$ denotes the set of separation points of $T$, i.e. the points $\gamma_1 = 0, \gamma_2, \ldots, \gamma_s$ (which are the left boundaries of the semi-intervals $I_1, \ldots, I_s$).

\begin{theorem}
\label{theo:birauzy1}
Let $T$ be a regular $s$-interval exchange transformation on $[\ell, r[$.
For any admissible semi-interval $I = [u,v[$, the transformation $S$ induced by $T$ on $I$ is a regular $s$-interval exchange transformation with separation points $\Sep(S) = \Div(I,T)\cap I$.
\end{theorem}
\begin{proof}
Since $T$ is regular, it is minimal.
Thus for each $i \in \{2, \ldots, s\}$ there are points $x_i, y_i \in ]u,v[$ such that there is a path from $x_i$ to $y_i$ passing by $\gamma_i$ but not containing any point of $I$ except at its origin and its end.
Since $T$ is regular, the $x_i$ are all distinct and the $y_i$ are all distinct.

Since $I$ is admissible, there exist $g, d \in \{1, \ldots, s\}$ such that $u \in N_{I,T}(\gamma_g)$ and $v \in N_{I,T}(\gamma_d)$.
Moreover,since $u$ is a neighbor of $\gamma_g$ with respect to $I$, $u$ is on the path from $x_g$ to $y_g$ (it can be either before or after $\gamma_g$).
Similarly, $v$ is on the path from $x_d$ to $y_d$ (see Figure~\ref{fig:biinduced} where $u$ is before $\gamma_g$ and $v$ is after $\gamma_d$).

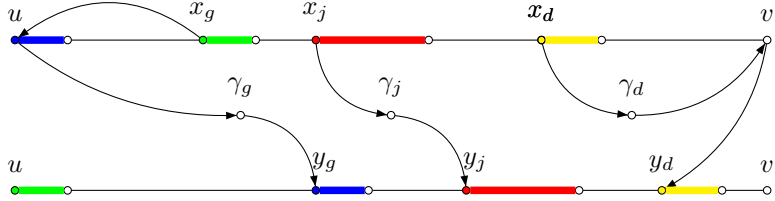
\begin{figure}[hbt]
\centering
\gasset{Nw=1,Nh=1,ExtNL=y,NLdist=2}
\begin{picture}(100,25)
\node[fillcolor=blue](uh)(0,20){$u$}
\node(x_{i_2})(7,20){}
\node[fillcolor=green](x_g)(25,20){$x_g$}
\node(x_{i_{k+1}})(32,20){}
\node[fillcolor=red](x_{i_j})(40,20){$x_j$}
\node(x_{i_{j+1}})(55,20){}
\node[fillcolor=yellow](x_d)(70,20){$x_d$}
\node(x_e)(78,20){}
\node(vh)(100,20){$v$}

\node(gamma_g)(30,10){$\gamma_g$}
\node(gamma_{i_j})(50,10){$\gamma_j$}
\node(gamma_d)(82,10){$\gamma_d$}

\node[fillcolor=green](ub)(0,0){$u$}
\node(y_{i_2})(7,0){}
\node[fillcolor=blue](y_g)(40,0){$\ \ y_g$}
\node(y_{i_{k+1}})(47,0){}
\node[fillcolor=red](y_{i_j})(60,0){$\ \ y_j$}
\node(y_{i_{j+1}})(75,0){}
\node(x_d)(70,20){$x_d$}
\node[fillcolor=yellow](y_d)(86,0){$y_d$}
\node(y_e)(94,0){}
\node(vb)(100,0){$v$}

\drawedge[AHnb=0,linecolor=blue,linewidth=1](uh,x_{i_2}){}
\drawedge[AHnb=0](x_{i_2},x_g){}
\drawedge[AHnb=0,linecolor=green,linewidth=1](x_g,x_{i_{k+1}}){}
\drawedge[AHnb=0](x_{i_{k+1}},x_{i_j}){}
\drawedge[AHnb=0,linecolor=red,linewidth=1](x_{i_j},x_{i_{j+1}}){}
\drawedge[AHnb=0](x_{i_{j+1}},x_d){}
\drawedge[AHnb=0,linecolor=yellow,linewidth=1](x_d,x_e){}
\drawedge[AHnb=0](x_e,vh){}
\drawedge[curvedepth=-5](x_g,uh){}
\drawedge[curvedepth=-3](uh,gamma_g){}
\drawedge[curvedepth=3](gamma_g,y_g){}
\drawedge[curvedepth=-3](x_{i_j},gamma_{i_j}){}
\drawedge[curvedepth=3](gamma_{i_j},y_{i_j}){}
\drawedge[curvedepth=-3](x_d,gamma_d){}
\drawedge[curvedepth=-3](gamma_d,vh){}
\drawedge[curvedepth=3](vh,y_d){}
\drawedge[AHnb=0,linecolor=green,linewidth=1](ub,y_{i_2}){}
\drawedge[AHnb=0](y_{i_2},y_g){}
\drawedge[AHnb=0,linecolor=blue,linewidth=1](y_g,y_{i_{k+1}}){}
\drawedge[AHnb=0](y_{i_{k+1}},y_{i_j}){}
\drawedge[AHnb=0,linecolor=red,linewidth=1](y_{i_j},y_{i_{j+1}}){}
\drawedge[AHnb=0](y_{i_{j+1}},y_d){}
\drawedge[AHnb=0,linecolor=yellow,linewidth=1](y_d,y_e){}
\drawedge[AHnb=0](y_e,vb){}
\end{picture}
\caption{The transformation induced on $[u, v[$.}
\label{fig:biinduced}
\end{figure}

Set $x_1 = y_1 = u$.
Let $(I_j)_{1 \leq j \leq s}$ be the partition of $I$ in semi-intervals such that $x_j$ is the left boundary of $I_j$ for $1 \leq j \leq s$.
Let $J_j$ be the partition of $I$ such that $y_j$ is the left boundary of $J_j$ for $1 \leq j \leq s$.
We will prove that
$$
S(I_j) =
\begin{cases}
J_j	&	\text{if $j \neq 1,g$} \\
J_1	&	\text{if $j = g$} \\
J_g	&	\text{if $j = 1$}
\end{cases}
$$
and that the restriction of $S$ to $I_j$ is a translation.

Assume first that $j \neq 1,g$.
Then $S(x_j) = y_j$.
Let $k$ be such that $y_j = T^k(x_j)$ and denote $I'_j = I_j \setminus x_j$.
We will prove by induction on $h$ that for $0 \leq h \leq k-1$, the set $T^h(I'_j)$ does not contain $u, v$ or any $x_i$.
It is true for $h = 0$.
Assume that it holds up to $h < k-1$.

For any $h'$ with $0 \leq h' \leq h$, the set $T^{h'}(I_j')$ does not contain any $\gamma_i$.
Indeed, otherwise there would exist $h''$ with $0 \leq h'' \leq h'$ such that $x_i \in T^{h''}(I'_j)$, a contradiction.
Thus $T$ is a translation on $T^{h'}(I_j)$.
This implies that $T^h$ is a translation on $I_j$.
Note also that $T^h(I'_j) \cap I = \emptyset$.
Assume the contrary.
We first observe that we cannot have $T^h(x_j) \in I$.
Indeed, $h < k$ implies that $T^h(x_j) \notin ]u, v[$.
And we cannot have $T^h(x_j) = u$ since $j \neq g$.
Thus $T^h(I'_j) \cap I \neq \emptyset$ implies that $u \in T^h(I'_j)$, a contradiction.

Suppose that $u = T^{h+1}(z)$ for some $z \in I'_j$.
Since $u$ is on the path from $x_g$ to $y_g$, it implies that for some $h'$ with $0 \leq h' \leq h$ we have $x_g = T^{h'}(z)$, a contradiction with the induction hypothesis.
A similar proof (using the fact that $v$ is on the path from $x_d$ to $y_d$) shows that $T^{h+1}(I'_j)$ does not contain $v$.
Finally suppose that some $x_i$ is in $T^{h+1}(I'_j)$.
Since the restriction of $T^h$ to $I_j$ is a translation, $T^h(I_j)$ is a semi-interval.
Since $T^{h+1}(x_j)$ is not in $I$ the fact that $T^{h+1}(I_j) \cap I$ is not empty implies that $u \in T^h(I_j)$, a contradiction.

This shows that $T^k$ is continuous at each point of $I'_j$ and that $S = T^k(x)$ for all $x \in I_j$.
This implies that the restriction of $S$ to $I_j$ is a translation into $J_j$.

If $j=1$, then $S(x_1) = S(u) = y_g$.
The same argument as above proves that the restriction of $S$ to $I_1$ is a translation form $I_1$ into $J_g$.
Finally if $j = g$, then $S(x_g) = x_1 = u$ and, similarly, we obtain that the restriction of $S$ to $I_g$ is a translation into $I_1$.

Since $S$ is the transformation induced by the transformation $T$ which is one to one, it is also one to one.
This implies that the restriction of $S$ to each of the semi-intervals $I_j$ is a bijection onto the corresponding interval $J_j,J_1$ or $J_g$ according to the value of $j$.

This shows that $S$ is an $s$-interval exchange transformation.
Since the orbits of the points $x_2, \cdots, x_s$ relative to $S$ are included in the orbits of $\gamma_2, \ldots, \gamma_s$, they are infinite and disjoint.
Thus $S$ is regular.

Let us  finally show that $\Sep(S) = \Div(I,T) \cap I$.
We have $\Sep(S) = \{x_1, x_2, \ldots, x_s\}$ and $x_i \in N_{I,T}(\gamma_i)$.
Thus $\Sep(S) \subset \Div(I,T) \cap I$.
Conversely, let $x \in \Div(I,T) \cap I$.
Then $x \in N_{I,T}(\gamma_i) \cap I$ for some $1 \leq i \leq s$.
If $i\ne 1,g$, then $x = x_i$.
If $i = 1$, then either $x = u$ (if $u = \ell$) or $x = x_{\pi(1)}$ since $\gamma_1 = T(\gamma_{\pi(1)})$.
Finally, if $i = g$ then $x = u$ or $x = x_g$.
Thus $x \in \Sep(S)$ in all cases.
\end{proof}

We have already noted that for any $s$-interval exchange transformation on $[\ell, r[$ and any semi-interval $I$ of $[\ell, r[$, the transformation $S$ induced by $T$ on $I$ is an interval exchange transformation on at most $s+2$-intervals.
Actually, it follows from the proof of Lemma 2, page 128 in~\cite{CornfeldFominSinai1982} that, if $T$ is regular and $S$ is an $s$-interval exchange transformation with separation points $\Sep(S) = \Div(I,T)\cap I$, then $I$ is admissible.
Thus the converse of Theorem~\ref{theo:birauzy1} is also true.

\subsection{A closure property}
In the following we will prove a closure property of the family of regular interval exchange sets.
The same property holds for Sturmian sets (see~\cite{JustinVuillon2000}) and for uniformly recurrent tree sets (see~\cite{BertheDeFeliceLeroyPerrinReutenauerRindone2014}).

\begin{lemma}
\label{lem:returnsind}
Let $T$ be a regular interval exchange transformation and let $F = F(T)$. 
For $w \in F$, let $S$ be the transformation induced by $T$ on $J_w$.
One has $x\in \RR_F(w)$ if and only if
$$
\Sigma_T(z) = x\Sigma_T(S(z))
$$
for some $z \in J_w$.
\end{lemma}
\begin{proof}
Assume first that $x\in \RR_F(w)$.
Then for any $z \in J_{w} \cap I_x$, we have $S(z) = T^{|x|}(z)$ and
$$
\Sigma_T(z) = x\Sigma_T(T^{|x|}(z)) = x\Sigma_T(S(z)).
$$
Conversely, assume that $\Sigma_T(z) = x\Sigma_T(S(z))$ for some $z \in J_w$.
Then $T^{|x|}(z) \in J_w$ and thus $wx \in A^*w$ which implies that $x \in \Gamma_F(w)$.
Moreover $x$ does not have a proper prefix in $\Gamma_F(w)$ and thus $x\in \RR_F(w)$.
\end{proof}

Since a regular interval exchange set is recurrent, the previous lemma says that the natural coding of a point in $J_w$ is a concatenation of first return words to $w$.
Moreover, note also that $T^n(z) \in J_w$ if and only if the prefix of length $n$ of $\Sigma_T(z)$ is a return word to $w$.

\begin{theorem}
\label{theo:returns}
Any derived set of a regular $s$-interval exchange set is a regular $s$-interval exchange set.
\end{theorem}
\begin{proof}
Let $T$ be a regular $s$-interval exchange transformation and let $F = F(T)$.

Let $w \in F$.
Since the semi-interval $J_w$ is admissible according to Proposition~\ref{pro:jadm}, the transformation $S$ induced by $T$ on $J_w$ is, by Theorem~\ref{theo:birauzy1}, an $s$-interval exchange transformation.
The corresponding partition of $J_w$ is the family $\left( J_{wx} \right)_{w \in \RR_F(w)}$.

Using Lemma~\ref{lem:returnsind} and the observation following, it is clear that $\Sigma_T(z) = f(\Sigma_S(z))$, where $z$ is a point of $J_w$ and $f : A^* \to \RR_F(w)^*$ is a coding morphism for $\RR_F(w)$.

Set $x = \Sigma_T(T^{-|w|}(z))$ and $y = \Sigma_T(z)$.
Then $x = wy$ and thus $\Sigma_S(z) = \D_f(x)$.
By Proposition~\ref{pro:derived}, this shows that the derived set of $F$ with respect to $f$ is $F(S)$.
\end{proof}

Theorem~\ref{theo:returns} implies, in particular, a result of~\cite{Vuillon2007}, i.e., that $\Card(\RR_F(w)) = \Card(A)$
(see also~\cite{BalkovaPelantovaSteiner2008} and~\cite{BertheDeFeliceLeroyPerrinReutenauerRindone2014}).

\section{Rauzy induction}
\label{sec:rauzy}

In this section we describe the transformation called Rauzy induction defined in~\cite{Rauzy1979} which operates on regular interval transformations and recall the results concerning this transformation (Theorems~\ref{theo:rauzy1} and \ref{theo:rauzy2}).
We introduce the definition of admissibility for an interval.
It generalizes in a natural way the notion of admissibility defined in~\cite{Rauzy1979}.
We also introduce a branching version of this transformation and generalize Rauzy's results to the two-sided case (Theorems~\ref{theo:birauzy1} and~\ref{theo:birauzy2}).
In particular we characterize in Theorem~\ref{theo:birauzy2} the admissible semi-intervals for an interval exchange transformation.

\subsection{One-sided Rauzy induction}
Let $T = T_{\lambda,\pi}$ be a regular $s$-interval exchange transformation on $[\ell, r[$.
Set $Z(T) = [\ell, \max\{\gamma_{s}, \delta_{\pi(s)}\}[$.

Note that $Z(T)$ is the largest semi-interval which is right-admissible for $T$.
We denote by $\psi(T)$ the transformation induced by $T$ on $Z(T)$.

The following result is Theorem 23 in~\cite{Rauzy1979}.

\begin{theorem}[Rauzy]
\label{theo:rauzy2}
Let $T$ be a regular interval exchange transformation.
A semi-interval $I$ is right admissible for $T$ if and only if there is an integer $n \geq 0$ such that $I = Z(\psi^n(T))$.
In this case, the transformation induced by $T$ on $I$ is $\psi^{n+1}(T)$.
\end{theorem}

The map $T \mapsto \psi(T)$ is called the \emph{right Rauzy induction}.
There are actually two cases according to $\gamma_{s} < \delta_{\pi(s)}$ (Case 0) or $\gamma_{s} > \delta_{\pi(s)}$ (Case 1).
We cannot have $\gamma_{s} = \delta_{\pi(s)}$ since $T$ is regular.

In Case 0, we have $Z(T) = [\ell, \delta_{\pi(s)}[$ and for any $z\in Z(T)$,
$$
S(z) =
\begin{cases}
T^2(z)	&	\text{if $z \in I_{a_{\pi(s)}}$} \\ 
T(z)	&	\text{otherwise}.
\end{cases}
$$
The transformation $S$ is the interval exchange transformation relative to $(K_a)_{a \in A}$ with $K_a = I_a\cap Z(T)$ for all $a \in A$.
Note that $K_a = I_a$ for $a \neq a_s$.
The translation values $\beta_a$ are defined as follows, denoting $\alpha_i, \beta_i$ instead of $\alpha_{a_i}, \beta_{a_i}$,
$$
\beta_i =
\begin{cases}
\alpha_{\pi(s)} + \alpha_s	&	\text{if $i = \pi(s)$} \\
\alpha_i			&	\text{otherwise.}
\end{cases}
$$
In summary, in Case 0, the semi-interval $J_{a_\pi(s)}$ is suppressed, the semi-interval $J_{a_s}$ is split into $S(K_{a_s})$ and $S(K_{a_{\pi(s)}})$.
The left boundaries of the semi-intervals $K_a$ are the left boundaries of the semi-intervals $I_a$.
The transformation is represented  in Figure~\ref{fig:rauzyind0}, in which the left boundary of the semi-interval $S(K_{a_{\pi(s)}})$ is denoted $\delta'_{\pi(s)}$.

\begin{figure}[hbt]
\centering
\gasset{Nw=1,Nh=1,ExtNL=y,NLdist=1,AHnb=0,ELside=r}
\begin{picture}(100,35)
\put(0,24){
\begin{picture}(100,15)
\node(0h)(0,9){$\ell$}
\node[fillcolor=blue](gamma_{pi(s)})(20,9){$\gamma_{\pi(s)}$}
\node(mu_{pi(s)})(30,9){}
\node[fillcolor=red](gamma_s)(80,9){$\gamma_s$}\node(1h)(100,9){$r$}
\drawedge(0h,gamma_{pi(s)}){}
\drawedge[linecolor=blue,linewidth=1](gamma_{pi(s)},mu_{pi(s)}){$a_{\pi(s)}$}
\drawedge(mu_{pi(s)},gamma_s){}
\drawedge[linecolor=red,linewidth=1](gamma_s,1h){$a_s$}

\node(0b)(0,0){}
\node[fillcolor=red](delta_s)(40,0){$\delta_s$}
\node(nu_s)(60,0){}
\node[fillcolor=blue](delta_{pi(s)})(90,0){$\delta_{\pi(s)}$}
\node(1b)(100,0){}
\drawedge(0b,delta_s){}
\drawedge[linecolor=red,linewidth=1](delta_s,nu_s){$a_s$}
\drawedge(nu_s,delta_{pi(s)}){}
\drawedge[linecolor=blue,linewidth=1](delta_{pi(s)},1b){$a_{\pi(s)}$}
\end{picture}
}
\put(50,13){\huge $\downarrow$}\put(90,13){$\vdots$}
\put(0,0){
\begin{picture}(100,15)
\node(0h)(0,9){$\ell$}
\node[fillcolor=blue](gamma_{pi(s)})(20,9){$\gamma_{\pi(s)}$}
\node(mu_{pi(s)})(30,9){}
\node[fillcolor=red](gamma_s)(80,9){$\gamma_s$}
\node(1h)(90,9){}
\drawedge(0h,gamma_{pi(s)}){}
\drawedge[linecolor=blue,linewidth=1](gamma_{pi(s)},mu_{pi(s)}){$a_{\pi(s)}$}
\drawedge(mu_{pi(s)},gamma_s){}
\drawedge[linecolor=red,linewidth=1](gamma_s,1h){$a_s$}

\node(0b)(0,0){}
\node[fillcolor=red](delta_{s,S})(40,0){$\delta_{s}$}
\node[fillcolor=blue](delta_{pi(s),S})(50,0){$\delta'_{\pi(s)}$}
\node(nu_{pi(s),S})(60,0){}
\node(delta_{pi(s)})(90,0){}
\drawedge(0b,delta_{s,S}){}
\drawedge[linecolor=red,linewidth=1](delta_{s,S},delta_{pi(s),S}){$a_s$}
\drawedge[linecolor=blue,linewidth=1](delta_{pi(s),S},nu_{pi(s),S}){$a_{\pi(s)}$}
\drawedge(nu_{pi(s),S},delta_{pi(s)}){}
\end{picture}
}
\end{picture}
\caption{Case 0 in Rauzy induction.}
\label{fig:rauzyind0}
\end{figure}

In Case 1, we have $Z(T) = [\ell, \gamma_{s}[$ and for any $z \in Z(T)$,
$$
S(z) =
\begin{cases}
T^2(z)	&	\text{if $z \in T^{-1}(I_{a_s})$} \\
T(z)	&	\text{otherwise}.
\end{cases}
$$
The transformation $S$ is the interval exchange transformation relative to  $(K_a)_{a\in A}$ with
$$
K_a =
\begin{cases}
T^{-1}(I_{a})		&	\text{if $a = a_s$} \\
T^{-1}(T(I_a)\cap Z(T))	&	\text{otherwise.}
\end{cases}
$$
Note that $K_a = I_a$ for $a \neq a_s$ and $a \neq a_{\pi(s)}$.
Moreover $K_a = S^{-1}(T(I_a) \cap Z(T))$ in all cases.
The translation values $\beta_i$ are defined by
$$
\beta_i =
\begin{cases}
\alpha_{\pi(s)} + \alpha_s	&	\text{if $i = s$} \\
\alpha_i			&	\text{otherwise.}
\end{cases}
$$
In summary, in Case 1, the semi-interval $I_{a_s}$ is suppressed, the semi-interval $I_{a_{\pi(s)}}$ is split into $K_{a_{\pi(s)}}$ and $K_{a_s}$.
The left boundaries of the semi-intervals $S(K_a)$ are the left boundaries of the semi-intervals $J_a$.
The transformation is represented in Figure~\ref{fig:rauzyind1}, where the left boundary of the semi-interval $K_{a_s}$ is denoted $\gamma'_{s}$.

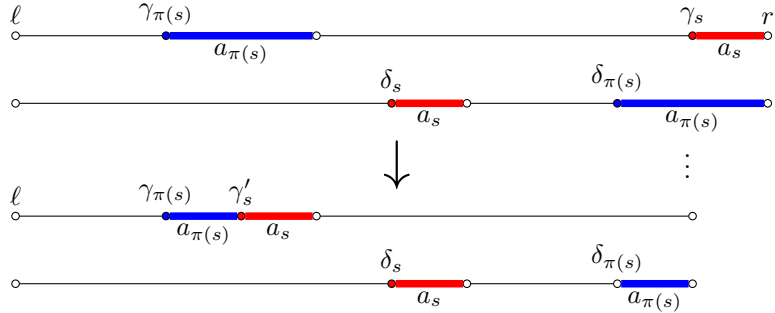
\begin{figure}[hbt]
\centering
\gasset{Nw=1,Nh=1,ExtNL=y,NLdist=1,AHnb=0,ELside=r}
\begin{picture}(100,35)
\put(0,24){
\begin{picture}(100,20)
\node(0h)(0,9){$\ell$}
\node[fillcolor=blue](gamma_{pi(s)})(20,9){$\gamma_{\pi(s)}$}
\node(mu_{pi(s)})(40,9){}
\node[fillcolor=red](gamma_s)(90,9){$\gamma_s$}\node(1h)(100,9){$r$}
\drawedge(0h,gamma_{pi(s)}){}
\drawedge[linecolor=blue,linewidth=1](gamma_{pi(s)},mu_{pi(s)}){$a_{\pi(s)}$}
\drawedge(mu_{pi(s)},gamma_s){}
\drawedge[linecolor=red,linewidth=1](gamma_s,1h){$a_s$}

\node(0b)(0,0){}
\node[fillcolor=red](delta_s)(50,0){$\delta_s$}
\node(nu_s)(60,0){}
\node[fillcolor=blue](delta_{pi(s)})(80,0){$\delta_{\pi(s)}$}
\node(1b)(100,0){}
\drawedge(0b,delta_s){}
\drawedge[linecolor=red,linewidth=1](delta_s,nu_s){$a_s$}
\drawedge(nu_s,delta_{pi(s)}){}
\drawedge[linecolor=blue,linewidth=1](delta_{pi(s)},1b){$a_{\pi(s)}$}
\end{picture}
}
\put(50,14){\huge $\downarrow$}\put(90,14){$\vdots$}
\put(0,0){
\begin{picture}(100,19)
\node(0h)(0,9){$\ell$}
\node[fillcolor=blue](gamma_{pi(s)})(20,9){$\gamma_{\pi(s)}$}
\node[fillcolor=red](gamma_{s,S})(30,9){$\gamma'_{s}$}
\node(mu_{pi(s)})(40,9){}
\node(1h)(90,9){}

\drawedge(0h,gamma_{pi(s)}){}
\drawedge[linecolor=blue,linewidth=1](gamma_{pi(s)},gamma_{s,S}){$a_{\pi(s)}$}
\drawedge[linecolor=red,linewidth=1](gamma_{s,S},mu_{pi(s)}){$a_s$}
\drawedge(mu_{pi(s)},gamma_s){}

\node(0b)(0,0){}
\node[fillcolor=red](delta_{s})(50,0){$\delta_{s}$}
\node(nu_{s})(60,0){}
\node(delta_{pi(s)})(80,0){$\delta_{\pi(s)}$}
\node(1b)(90,0){}

\drawedge(0b,delta_s){}
\drawedge[linecolor=red,linewidth=1](delta_{s},nu_{s}){$a_s$}
\drawedge(nu_s,delta_{pi(s)}){}
\drawedge[linecolor=blue,linewidth=1](delta_{pi(s)},1b){$a_{\pi(s)}$}
\end{picture}
}
\end{picture}
\caption{Case 1 in Rauzy induction.}
\label{fig:rauzyind1}
\end{figure}

\begin{example}
\label{ex:induced2}
Consider again the transformation $T$ of Example~\ref{ex:2alpha}.
Since $Z(T) = [0, 2\alpha[$, the transformation $\psi(T)$ is represented in Figure~\ref{fig:3ietind}.
The transformation $\psi^2(T)$ is represented in Figure~\ref{fig:3ietind2}.

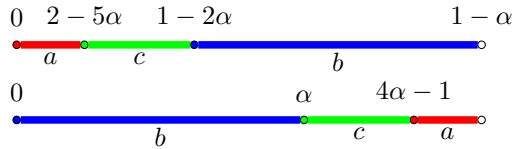
\begin{figure}[hbt]
\centering
\gasset{Nw=1,Nh=1,ExtNL=y,NLdist=2,AHnb=0,ELside=r}
\begin{picture}(70,20)
\node[fillcolor=red](0h)(0,10){$0$}
\node[fillcolor=green](2-5alpha)(9,10){$2-5\alpha$}
\node[fillcolor=blue](1-2alpha)(23.6,10){$1-2\alpha$}
\node(1-alphah)(61.8,10){$1-\alpha$}
\drawedge[linecolor=red,linewidth=1](0h,2-5alpha){$a$}
\drawedge[linecolor=green,linewidth=1](2-5alpha,1-2alpha){$c$}
\drawedge[linecolor=blue,linewidth=1](1-2alpha,1-alphah){$b$}

\node[fillcolor=blue](0b)(0,0){$0$}
\node[fillcolor=green](alpha)(38.2,0){$\alpha$}
\node[fillcolor=red](4alpha-1)(52.8,0){$4\alpha-1$}\node(1-alpha)(61.8,0){}
\drawedge[linecolor=blue,linewidth=1](0b,alpha){$b$}
\drawedge[linecolor=green,linewidth=1](alpha,4alpha-1){$c$}
\drawedge[linecolor=red,linewidth=1](4alpha-1,1-alpha){$a$}
\end{picture}
\caption{The transformation  $\psi^2(T)$.}
\label{fig:3ietind2}
\end{figure}
\end{example}

The symmetrical notion of \emph{left Rauzy induction} is defined similarly.

Let $T = T_{\lambda, \pi}$ be a regular $s$-interval exchange transformation on $[\ell, r[$.
Set $Y(T) = [\min\{\mu_1, \nu_{\pi(1)}\}, r[$.
We denote by $\varphi(T)$ the transformation induced by $T$ on $Y(T)$.
The map $T \mapsto \varphi(T)$ is called the \emph{left Rauzy induction}. 

Note that one has also $Y(T) = [\min\{\gamma_2, \delta_{\pi(2)}\}, r[$.

The symmetrical statements of Theorem~\ref{theo:rauzy2} also hold for left admissible intervals.

\subsection{Branching induction}
The following is a generalization of Theorem~\ref{theo:rauzy2}.

\begin{theorem}
\label{theo:birauzy2}
Let $T$ be a regular $s$-interval exchange transformation on $[\ell,r[$.
A semi-interval $I$ is admissible for $T$ if and only if there is a sequence $\chi \in \{\varphi, \psi\}^*$ such that $I$ is the domain of $\chi(T)$.
In this case, the transformation induced by $T$ on $I$ is $\chi(T)$.
\end{theorem}

We first prove the following lemmas, in which we assume that $T$ is a regular $s$-interval exchange transformation on $[\ell, r[$.
Recall that $Y(T), Z(T)$ are the domains of $\varphi(T), \psi(T)$ respectively.

\begin{lemma}
\label{lem:init}
If a semi-interval  $I$ strictly included in $[\ell, r[$ is admissible for $T$, then either $I \subset Y(T)$ or $I\subset Z(T)$.
\end{lemma}
\begin{proof}
Set $I = [u,v[$.
Since $I$ is strictly included in $[\ell, r[$, we have either $\ell < u$ or $v < r$.
Set $Y(T) = [y, r[$ and $Z(T) = [\ell, z[$.

Assume that $v < r$.
If $y \leq u$, then $I \subset Y(T)$.
Otherwise, let us show that $v \leq z$.
Assume the contrary.
Since $I$ is admissible, we have $v = T^k(\gamma_i)$ with $k \in E_{I,T}(\gamma_i)$ for some $i$ with $1 \leq i \leq s$.
But $k > 0$ is impossible since $u < T(\gamma_i) < v$ implies $T(\gamma_i) \in\ ]u,v[$, in contradiction with the fact that $k < \rho_I^+(\gamma_i)$.
Similarly, $k \leq 0$ is impossible since $u < \gamma_i < v$ implies $\gamma_i \in\ ]u,v[$.
Thus $I\subset Z(T)$.

The proof in the case $\ell < u$ is symmetric.
\end{proof}

The next lemma is  the two-sided version of Lemma 22 in~\cite{Rauzy1979}.

\begin{lemma}
\label{lem:bi22}
Let $T$ be a regular $s$-interval exchange transformation on $[\ell, r[$.
Let $J$ be an admissible semi-interval for $T$ and let $S$ be the transformation induced by $T$ on $J$.
A semi-interval $I \subset J$ is admissible for $T$ if and only if it is admissible for $S$.
Moreover $\Div(J,T) \subset \Div(I,T)$.
\end{lemma}
\begin{proof}
Set $J = [t,w[$ and $I = [u,v[$.
Since $J$ is admissible for $T$, the transformation $S$ is a regular $s$-interval exchange transformation by Theorem~\ref{theo:birauzy1}.

Suppose first that $I$ is admissible for $T$.
Then $u = T^g(\gamma_i)$ with $g \in E_{I,T}(\gamma_{i})$ for some $1 \leq i \leq s$, and $v = T^{d}(\gamma_{j})$ with $d \in E_{I,T}(\gamma_{j})$ for some $1 \leq j \leq s$ or $v = r$.

Since $S$ is the transformation induced by $T$ on $J$ there is a separation point $x$ of $S$ of the form $x = T^m(\gamma_{i})$ with $m = -\rho^-_{J,T}(\gamma_i)$ 
and thus $m \in E_{J,T}(\gamma_i)$.
Thus $u = T^{g-m}(x)$.

Assume first that $g-m > 0$.
Since $u, x \in J$, there is an integer $n$ with $0 < n \leq g-m$ such that $u = S^n(x)$.

Let us show that $n \in E_{I,S}(x)$.
Assume by contradiction that $\rho_ {I,S}^+(x) \leq n$.
Then there is some $k$ with $0 < k \leq n$ such that $S^k(x) \in ]u, v[$.
But we cannot have $k = n$ since $u \notin\ ]u,v[$.
Thus $k < n$.

Next, there is $h$ with $0 < h< g-m$ such that $T^h(x) = S^k(x)$.
Indeed, setting $y = S^k(x)$, we have $u = T^{g-m-h}(y) = S^{n-k}(y)$ and thus $h < g-m$.
If $0 < h \leq -m$, then $T^h(x) = T^{m+h}(\gamma_i) \in I\subset J$ contradicting the hypothesis that $m \in E_{J,T}(\gamma_i)$.
If $-m < h < g-m$, then $T^h(x) = T^{m+h}(\gamma_i) \in I$, contradicting the fact that $g \in E_{I,T}(\gamma_i)$.
This shows that $n \in E_{I,S}(x)$ and thus that $u \in \Div(I,S)$.

Assume next that $g-m \leq 0$.
There is an integer $n$ with $g-m \leq n \leq 0$ such that $u = S^n(x)$.
Let us show that $n \in E_{I,S}(x)$.
Assume by contradiction that $n < -\rho^-_{I,S}(x)$.
Then there is some $k$ with $n < k < 0$ such that $S^k(x) = T^h(x)$.
Then $T^h(x) = T^{h+m}(\gamma_i) \in I$ with $g < h+m < m$, in contradiction with the hypothesis that $m \in E_{I,T}(\gamma_i)$.

We have proved that $u \in \Div(I,S)$.
If $v = r$, the proof that $I$ is admissible for $S$ is complete.
Otherwise, the proof that $v \in \Div(I,S)$ is similar to the proof for $u$.

Conversely, if  $I$ is admissible for $S$, there is some $x \in \Sep(S)$ and $g \in E_{I,S}(x)$ such that $u = S^g(x)$.
But $x = T^m(\gamma_i)$ and since $u, x\in J$ there is some $n$ such that $u = T^n(\gamma_i)$.

Assume for instance that $n > 0$ and suppose that there exists $k$ with $0 < k < n$ such that $T^k(\gamma_i) \in ]u, v[$.
Then, since $I \subset J$, $T^k(\gamma_i)$ is of the form $S^h(x)$ with $0 < h < g$ which contradicts the fact that $g \in E_{I,S}(x)$.
Thus $n \in E_{I,T}(\gamma_i)$ and $u \in \Div(I,T)$.

The proof is similar in the case $n \leq 0$.

If $v = r$, we have proved that $I$ is admissible for $T$.
Otherwise, the proof that $v \in \Div(I,T)$ is similar.

Finally, assume that $I$ is admissible for $T$ (and thus for $S$).
For any $\gamma_i \in \Sep(T)$, one has
$$
\rho^-_{I,T}(\gamma_i) \geq \rho^-_{J,T}(\gamma_i)
\quad \text{ and } \quad
\rho^+_{I,T}(\gamma_i) \geq \rho^+_{J,T}(\gamma_i)
$$
showing that $\Div(J,T) \subset \Div(I,T)$.
\end{proof}

The last lemma is the key argument to prove Theorem~\ref{theo:birauzy2}.
It is a tree version of the argument used by Rauzy in~\cite{Rauzy1979}.

\begin{lemma}
\label{lem:finiteness}
For any admissible interval $I \subset[\ell, r[$, the set $\F$ of sequences $\chi \in \{\varphi, \psi\}^*$ such that $I \subset D(\chi(T))$ is finite.
\end{lemma}
\begin{proof}
The set $\F$ is suffix-closed.
Indeed it contains the empty word because $[\ell, r[$ is admissible.
Moreover, for any $\xi, \chi \in \{\varphi,\psi\}^*$, one has $D(\xi \chi(T)) \subset D(\chi(T))$ and thus $\xi \chi \in \F$ implies $\chi \in \F$. 

The set $\F$ is finite.
Indeed, by Lemma~\ref{lem:bi22}, applied to $J = D(\chi(T))$, for any $\chi \in \F$, one has $\Div(D(\chi(T)),T) \subset \Div(I,T)$.
In particular,the boundaries of $D(\chi(T))$ belong to $\Div(I,T)$.
Since $\Div(I,T)$ is a finite set, this implies that there is a finite number of possible semi-intervals $D(\chi(T))$.
Thus there is is no infinite word with all its suffixes in $\F$.
Since the sequences $\chi$ are binary, this implies that $\F$ is finite.
\end{proof}

\begin{proofof}{ of Theorem~\ref{theo:birauzy2}}
We first prove by induction on the length of $\chi$ that the domain $I$ of $\chi(T)$ is admissible and that the transformation induced by $T$ on $I$ is $\chi(T)$.
It is true for $|\chi| = 0$ since $[\ell, r[$ is admissible and $\chi(T) = T$.
Next, assume that $J = D(\chi(T))$ is admissible and that the transformation induced by $T$ on $J$ is $\chi(T)$.
Then $D(\varphi \chi(T))$ is admissible for $\chi(T)$ since $D(\varphi \chi(T)) = Y(\chi(T))$.
Thus $I = D(\varphi\chi(T))$ is admissible for $T$ by Lemma~\ref{lem:bi22} and the transformation induced by $T$ on $I$ is $\varphi \chi(T)$.
The same proof holds for $\psi \chi$.

Conversely, assume that $I$ is admissible.
By Lemma~\ref{lem:finiteness}, the set $\F$ of sequences $\chi \in \{\varphi, \psi\}^*$ such that $I\subset D(\chi(T))$ is finite. 

Thus there is some $\chi \in \F$ such that $\varphi \chi, \psi \chi \notin \F$.
If $I$ is strictly included in $D(\chi(T))$, then by Lemma~\ref{lem:init} applied to $\chi(T)$, we have $I\subset Y(\chi(T)) = D(\varphi\chi(T))$ or $I\subset Z(\chi(T))=D(\psi\chi(T))$, a contradiction.
Thus $I = D(\chi(T))$.
\end{proofof}

We close this subsection with a result concerning the dynamics of the branching induction.

\begin{theorem}
\label{theo:length}
For any sequence $(T_n)_{n \geq 0}$ of regular interval exchange transformations such that $T_{n+1} = \varphi(T_n)$ or $T_{n+1} = \psi(T_n)$ for all $n \geq 0$, the length of the domain of $T_n$ tends to $0$ when $n \rightarrow \infty$.
\end{theorem}
\begin{proof}
Assume the contrary and let $I$ be an open interval included in the domain of $T_n$ for all $n \geq 0$.
The set $\Div(I,T) \cap I$ is formed of $s$ points.
For any pair $u, v$ of consecutive elements of this set, the semi-interval $[u, v[$ is admissible.
By Lemma~\ref{lem:finiteness}, there is an integer $n$ such that the domain of $T_n$ does not contain $[u, v[$, a contradiction.
\end{proof}

\subsection{Equivalence relation}
Let $[\ell_1, r_1[$, $[\ell_2, r_2[$ be two semi-intervals of the real line.
Let $T_1 = T_{\lambda_1, \pi_1}$ be an $s$-interval exchange transformation relative to a partition of $[\ell_1, r_1[$ and $T_2 = T_{\lambda_2, \pi_2}$ another $s$-interval exchange transformations relative to $[\ell_2, r_2[$.
We say that $T_1$ and $T_2$ are \emph{equivalent} if $\pi_1 = \pi_2$ and $\lambda_1 = c \lambda_2$ for some $c > 0$.
Thus, two interval exchange transformations are equivalent if we can obtain the second from the first by a rescaling following by a translation.
We denote by $\left[T_{\lambda, \pi}\right]$ the equivalence class of $T_{\lambda, \pi}$.

\begin{example}
\label{ex:simphi6T}
Let $S = T_{\mu,\pi}$ be the $3$-interval exchange transformation on a partition of the semi-interval $[2\alpha, 1[$, with $\alpha= (3-\sqrt{5})/2$, represented in Figure~\ref{fig:simphi6T}.
$S$ is equivalent to the transformation $T = T_{\lambda,\pi}$ of Example~\ref{ex:2alpha}, with length vector $\lambda = \left( 1-2\alpha, \alpha, \alpha \right)$ and permutation the cycle $\pi = (132)$.
Indeed the length vector $\mu = \left( 8\alpha-3, 2-5\alpha, 2-5\alpha \right)$ satisfies $\mu = \frac{2-5\alpha}{\alpha} \lambda$.

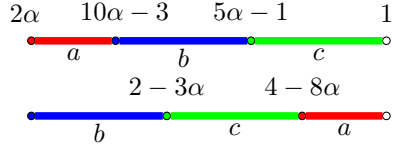
\begin{figure}[hbt]
\centering
\gasset{Nw=1,Nh=1,ExtNL=y,NLdist=2,AHnb=0,ELside=r}
\begin{picture}(47.2,20)
\node[fillcolor=red](2alphah)(0,10){$2\alpha \, \, \,$}
\node[fillcolor=blue](10alpha-3h)(11.2,10){$\, \, \, \, 10\alpha-3$}
\node[fillcolor=green](5alpha-1h)(29.2,10){$5\alpha-1$}
\node(1h)(47.2,10){$1$}
\drawedge[linecolor=red,linewidth=1](2alphah,10alpha-3h){$a$}
\drawedge[linecolor=blue,linewidth=1](10alpha-3h,5alpha-1h){$b$}
\drawedge[linecolor=green,linewidth=1](5alpha-1h,1h){$c$}

\node[fillcolor=blue](2alphab)(0,0){}
\node[fillcolor=green](2-3alphab)(18,0){$2-3\alpha$}
\node[fillcolor=red](4-8alphab)(36,0){$4-8\alpha$}
\node(1b)(47.2,0){}
\drawedge[linecolor=blue,linewidth=1](2alphab,2-3alphab){$b$}
\drawedge[linecolor=green,linewidth=1](2-3alphab,4-8alphab){$c$}
\drawedge[linecolor=red,linewidth=1](4-8alphab,1b){$a$}
\end{picture}
\caption{The transformation $S$.}
\label{fig:simphi6T}
\end{figure}
\end{example}

Note that if $T$ is a minimal (resp. regular) interval exchange transformation and $[S] = [T]$, then $S$ is also minimal (resp. regular).

For an interval exchange transformation $T$ we consider the directed labeled graph $\mathcal{G}(T)$, called the \emph{induction graph} of $T$, defined as follows.
The vertices are the equivalence classes of transformations obtained starting from $T$ and applying all possible $\chi \in \left\{ \psi, \varphi \right\}^*$.
There is an edge labeled $\psi$ (resp. $\varphi$) from a vertex $[S]$ to a vertex $[U]$ if and only if $U = \psi(S)$ (resp $\varphi(S)$) for two transformations $S \in [S]$ and $U \in [U]$.

\begin{example}
\label{ex:equivalencegraph}
Let $\alpha = \frac{3-\sqrt{5}}{2}$ and $R$ be a rotation of angle $\alpha$. By Example~\ref{ex:rotation}, $R$ is a $2$-interval exchange transformation on $[0,1[$ relative to the partition $[0, 1~-~\alpha[$, $[1-\alpha, 1[$.
The induction graph $\mathcal{G}(R)$ of the transformation is represented in the left of Figure~\ref{fig:graphrotation}.
\end{example}

Note that for a $2$-interval exchange transformation $T$, one has $[\psi(T)] = [\varphi(T)]$, whereas in general the two transformations are not equivalent.

The induction graph of an interval exchange transformation can be infinite.
A sufficient condition for the induction graph to be finite is given in Section~\ref{sec:quadratic}.
$$ $$
Let now introduce a variant of this equivalence relation (and of the related graph).
We consider the case of two transformation ``equivalent'' up to reflection (and up to the separation points).
In Section~\ref{subsec:nc} we will justify this choice in terms of natural coding of two specular points.

For an $s$-interval exchange transformation $T = T_{\lambda, \pi}$, with length vector $\lambda = \left( \lambda_1, \lambda_2, \ldots, \lambda_s \right)$, we define the \emph{mirror transformation} $\widetilde{T} = T_{\widetilde{\lambda}, \tau \circ \pi}$ of $T$, where $\widetilde{\lambda} = \left( \lambda_s, \lambda_{s-1}, \ldots, \lambda_1 \right)$ and $\tau : i \mapsto (s-1+1)$ is the permutation that reverses the names of the semi-intervals.

Given two interval exchange transformations $T_1$ and $T_2$ on the same alphabet relative to two partitions of two semi-intervals $[\ell_1, r_1[$ and $[\ell_2, r_2[$ respectively, we say that $T_1$ and $T_2$ are \emph{similar} either if $[T_1] = [T_2]$ or $[T_1] = [\widetilde{T_2}]$.
Clearly, similarity is also an equivalent relation.
We denote by $\langle T \rangle$ the class of transformations similar to $T$.

\begin{example}
\label{ex:phi6T}
Let $T$ be the interval exchange transformation of Example~\ref{ex:2alpha}.
The transformation $U = \varphi^6(T)$ is represented in Figure~\ref{fig:phi6T} (see also Example~\ref{ex:u}).
It is easy to verify that $U$ is similar to the transformation $S$ of Example~\ref{ex:simphi6T}.
Indeed, we can obtain the second transformation (up to the separation points and the end points) by taking the mirror image of the domain.

Note that the order of the labels, i.e. the order of the letters of the alphabet, may be different from the order of the original transformation.

\begin{figure}[hbt]
\centering
\gasset{Nw=1,Nh=1,ExtNL=y,NLdist=2,AHnb=0,ELside=r}
\begin{picture}(47.2,20)
\node[fillcolor=blue](2alphah)(0,10){$2\alpha$}
\node[fillcolor=red](2-3alphah)(18,10){$2-3\alpha$}
\node[fillcolor=green](4-8alphah)(36,10){$4-8\alpha$}
\node(1h)(47.2,10){$1$}
\drawedge[linecolor=blue,linewidth=1](2alphah,2-3alphah){$b$}
\drawedge[linecolor=red,linewidth=1](2-3alphah,4-8alphah){$a$}
\drawedge[linecolor=green,linewidth=1](4-8alphah,1h){$c$}

\node[fillcolor=green](2alphab)(0,0){}
\node[fillcolor=blue](10alpha-3b)(11.2,0){$10\alpha-3$}
\node[fillcolor=red](5alpha-1b)(29.2,0){$5\alpha-1$}
\node(1b)(47.2,0){}
\drawedge[linecolor=green,linewidth=1](2alphab,10alpha-3b){$c$}
\drawedge[linecolor=blue,linewidth=1](10alpha-3b,5alpha-1b){$b$}
\drawedge[linecolor=red,linewidth=1](5alpha-1b,1b){$a$}
\end{picture}
\caption{The transformation $U$.}
\label{fig:phi6T}
\end{figure}
\end{example}

As of the equivalence relation, also similarity preserves minimality and regularity.

Let $T$ be an interval exchange transformation.
We denote by
$$\mathcal{S}(T) = \bigcup_{n \in \Z} T^n \big(\Sep(T)\big)$$
the union of the orbits of the separation points.
Let $S$ be an interval exchange transformation similar to $T$.
Thus, there exists a bijection $f : D(T) \setminus \mathcal{S}(T) \to D(S) \setminus \mathcal{S}(S)$. This bijection is given by an affine transformation, namely a rescaling following by a translation if $T$ and $S$ are equivalent and a rescaling following by a translation and a reflection otherwise.
By the previous remark, if $T$ is a minimal exchange interval transformation and $S$ is similar to $T$, then the two interval exchange sets $F(T)$ and $F(S)$ are equal up to permutation, that is there exists a permutation $\pi$ such that one for every $w=a_0 a_1 \cdots a_{n-1} \in F(T)$ there exists a unique word $v = b_0 b_1 \cdots b_{n-1} \in F(S)$ such that $b_i = \pi(a_i)$ for all $i = 1, 2, \ldots n-1$.

In a similar way as before, we can use the similarity in order to construct a graph.
For an interval exchange transformation $T$ we define $\widetilde{\mathcal{G}}(T)$ the \emph{modified induction graph} of $T$ as the directed (unlabeled) graph with vertices the similar classes of transformations obtained starting from $T$ and applying all possible $\chi \in \left\{ \psi, \varphi \right\}^*$ and an edge from $\langle S \rangle$ to $\langle U \rangle$ if $U=\psi(S)$ or $U=\varphi(S)$ for two transformations $S \in \langle S \rangle$ and $U \in \langle U \rangle$.

Note that this variant appears naturally when considering the Rauzy induction of a $2$-interval exchange transformation as a continued fraction expansion.
There exists a natural bijection between the closed interval $[0,1]$ of the real line and the set of $2$-interval exchange transformation given by the map $x \mapsto T_{\lambda, \pi}$ where $\pi = (12)$ and $\lambda = \left( \lambda_1, \lambda_2 \right)$ is the length vector such that $x = \frac{\lambda_1}{\lambda_2}$.

In this view, the Rauzy induction corresponds to the Euclidean algorithm (see~\cite{MiernowskiNogueira2013} for more details), i.e. the map $\mathcal{E} : \R_+^2 \to \R_+^2$ given by
\begin{displaymath}
\mathcal{E}(\lambda_1, \lambda_2) =
\begin{cases}
\left( \lambda_1 - \lambda_2, \lambda_2 \right)
&\text{ if $\lambda_1 \geq \lambda_2$}\\
\left( \lambda_1, \lambda_2 - \lambda_2 \right)
&\text{otherwise}.
\end{cases}
\end{displaymath}

Applying iteratively the Rauzy induction starting from $T$ corresponds then to the continued fraction expansion of $x$.

\begin{example}
Let $\alpha$ and $R$ be as in Example~\ref{ex:equivalencegraph}.
The modified induction graph $\widetilde{\mathcal{G}}(R)$ of the transformation is represented on the right of Figure~\ref{fig:graphrotation}.
Note that the ratio of the two lengths of the semi-intervals exchanged by $T$ is $\frac{1-\alpha}{\alpha} = \frac{1+\sqrt{5}}{2} = \phi = 1 + \frac{1}{1+\frac{1}{1 + \cdots}}$.

\begin{figure}[hbt]
\centering
\gasset{Nframe=y}
\centering
\gasset{Nframe=y}
\begin{picture}(60,14)(0,-7)
\node(R1)(0,0){$[R]$}
\node(R2)(20,0){}
\drawedge[curvedepth=3](R1,R2){$\psi, \varphi$}
\drawedge[curvedepth=3](R2,R1){$\psi, \varphi$}

\node(R)(55,0){$\langle R \rangle$}
\drawloop[curvedepth=3,loopangle=0](R){}
\end{picture}
\caption{Induction graph and modified induction graph of the rotation $R$ of angle $\alpha$= $(3-\sqrt{5})/2$.}
\label{fig:graphrotation}
\end{figure}
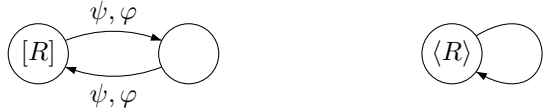
\end{example}

\subsection{Induction and automorphisms}
Let $T = T_{\lambda,\pi}$ be a regular interval exchange on $[\ell,r[$ relative to $(I_a)_{a \in A}$.
Set $A=\{a_1,\ldots,a_s\}$.
Recall now from Subsection~\ref{subsec:nc} that for any $z \in [\ell,r[$, the natural coding of $T$ relative to $z$ is the infinite word $\Sigma_T(z) = b_0 b_1 \cdots$ on the alphabet $A$ with $b_n \in A$ defined for $n \geq 0$ by $b_n = a$ if $T^n(z) \in I_a$.

Denote by $\theta_1$, $\theta_2$ the morphisms from $A^*$ into itself defined by
$$
\theta_1(a) =
\begin{cases}
a_{\pi(s)}a_s	&	\text{if $a = a_{\pi(s)}$} \\
a		&	\text{otherwise}
\end{cases},
\qquad
\theta_2(a) =
\begin{cases}
a_{\pi(s)}a_s	&	\text{if $a=a_s$} \\
a		&	\text{otherwise}
\end{cases}.
$$
The morphisms $\theta_1, \theta_2$ extend to automorphisms of the free group on $A$.

The following result already appears in~\cite{Jullian2012}.
We give a proof for the sake of completeness.

\begin{proposition}
\label{pro:autoelem}
Let $T$ be a regular interval exchange transformation on the alphabet $A$ and let $S = \psi(T)$, $I = Z(T)$.
There exists an automorphism $\theta$ of the free group on $A$ such that $\Sigma_T(z) = \theta(\Sigma_S(z))$ for any $z \in I$.
\end{proposition}
\begin{proof}
Assume first that $\gamma_{s} < \delta_{\pi(s)}$ (Case 0).
We have $Z(T) = [\ell, \delta_{\pi(s)}[$ and for any $x \in Z(T)$,
$$
S(z) =
\begin{cases}
T^2(z)	&	\text{if $z \in K_{a_{\pi(s)}} = I_{a_{\pi(s)}}$} \\ 
T(z)	&	\text{otherwise}.
\end{cases}
$$
We will prove by induction on the length of $w$ that for any $z \in I$, $\Sigma_S(z) \in wA^*$ if and only if $\Sigma_T(z) \in \theta_1(w)A^*$.
The property is true if $w$ is the empty word.
Assume next that $w = av$ with $a \in A$ and thus that $z \in I_a$.
If $a \neq a_{\pi(s)}$, then $\theta_1(a) = a$, $S(z) = T(z)$ and
$$
\Sigma_S(z) \in avA^* \Leftrightarrow \Sigma_S(S(z)) \in vA^* \Leftrightarrow \Sigma_T(T(z) )\in \theta_1(v)A^* \Leftrightarrow \Sigma_T(z) \in \theta_1(w)A^*.
$$
Otherwise, $\theta_1(a) = a_{\pi(s)}a_s$, $S(z) = T^2(z)$.
Moreover, $\Sigma_T(z) = a_{\pi(s)}a_s\Sigma_T(T^2(z))$ and thus
$$
\Sigma_S(z) \in avA^* \Leftrightarrow \Sigma_S(S(z)) \in vA^* \Leftrightarrow \Sigma_T(T^2(z)) \in \theta_1(v)A^* \Leftrightarrow \Sigma_T(z) \in \theta_1(w)A^*.
$$
If $\delta_{\pi(s)} < \gamma_{s}$ (Case 1), we have $Z(T) = [\ell,\gamma_{s}[$ and for any $z \in Z(T)$,
$$
S(z) =
\begin{cases}
T^2(z)	&	\text{if $z \in K_{a_s} = T^{-1}(I_{a_s})$} \\
T(z)	&	\text{otherwise}.
\end{cases}
$$
As in Case 0, we will prove by induction on the length of $w$ that for any $z \in I$, $\Sigma_S(z) \in wA^*$ if and only if $\Sigma_T(z) \in \theta_2(w)A^*$.

The property is true if $w$ is empty.
Assume next that $w = av$ with $a \in A$.
If $a \neq a_s$, then $\theta_2(a) = a$, $S(z) = T(z)$ and $z \in K_a \subset I_a$.
Thus
$$
\Sigma_S(z) \in avA^* \Leftrightarrow \Sigma_S(S(z)) \in vA^* \Leftrightarrow \Sigma_T(T(z)) \in \theta_2(v)A^* \Leftrightarrow \Sigma_T(z) \in \theta_2(w)A^*.
$$
Next, if $a = a_s$, then $\theta_2(a) = a_{\pi(s)}a_s$, $S(z) = T^2(z)$ and $z \in K_{a_s} = T^{-1}(I_{a_s}) \subset I_{a_{\pi(s)}}$.
Thus
$$
\Sigma_S(z) \in avA^*\Leftrightarrow \Sigma_S(S(z)) \in vA^* \Leftrightarrow \Sigma_T(T^2(z)) \in \theta_2(v)A^* \Leftrightarrow \Sigma_T(z) \in \theta_2(w)A^*.
$$
where the last equivalence results from the fact that $\Sigma_T(z) \in a_{\pi(s)}a_sA^*$.
This proves that $\Sigma_T(z) = \theta_2(\Sigma_S(z))$.
\end{proof}

\begin{example}
Let $T$ be the transformation of Example~\ref{ex:2alpha}.
The automorphism $\theta_1$ is defined by
$$
\theta_1(a) = ac, \quad \theta_1(b) = b, \quad \theta_1(c) = c.
$$
The right Rauzy induction gives the transformation $S = \psi(T)$ computed in Example~\ref{ex:induced1}.
One has $\Sigma_S(\alpha) = bacba \cdots$ and $\Sigma_T(\alpha) = baccbac \cdots = \theta_1(\Sigma_S(\alpha))$.
\end{example}

We state the symmetrical version of Proposition~\ref{pro:autoelem} for left Rauzy induction.
The proof is analogous.

\begin{proposition}
\label{pro:autoelemsym}
Let $T$ be a regular interval exchange transformation on the alphabet $A$ and let $S = \varphi(T)$, $I = Y(T)$.
There exists an automorphism $\theta$ of the free group on $A$ such that $\Sigma_T(z) = \theta(\Sigma_S(z))$ for any $z \in I$.
\end{proposition}

Combining Propositions~\ref{pro:autoelem} and~\ref{pro:autoelemsym}, we obtain the following statement.

\begin{theorem}
\label{theo:inductionbi}
Let $T$ be a regular interval exchange transformation.
For $\chi \in \{\varphi,\psi\}^*$, let $S = \chi(T)$ and let $I$ be the domain of $S$.
There exists an automorphism $\theta$ of the free group on $A$ such that $\Sigma_T(z) = \theta(\Sigma_S(z))$ for all $z \in I$.
\end{theorem}
\begin{proof}
The proof follows easily by induction on the length of $\chi$ using Propositions~\ref{pro:autoelem} and~\ref{pro:autoelemsym}.
\end{proof}

Note that if the transformations $T$ and $S = \chi(T)$, with $\chi \in \left\{\psi, \varphi\right\}^*$ , are equivalent, then there exists a point $z_0 \in D(S) \subseteq D(T)$ such that $z_0$ is a fixed point of the isometry that transforms $D(S)$ into $D(T)$ (if $\chi$ is different from the identity map, this point is unique).
In that case one has $\Sigma_S (z_0) = \Sigma_T (z_0) = \theta\left(\Sigma_S (z_0)\right)$ for an appropriate automorphism $\theta$, i.e. $\Sigma_T (z_0)$ is a fixed point
of an appropriate automorphism.

\begin{corollary}
Let $T$ be a regular interval exchange transformation.
For $w \in F(T)$, the set $\RR_F(w)$ is a basis of the free group on $A$.
\end{corollary}
\begin{proof}
By Proposition~\ref{pro:jadm}, the semi-interval $J_w$ is admissible.
By Theorem~\ref{theo:birauzy2} there is a sequence $\chi \in \{\varphi,\psi\}^*$ such that $D(\chi(T)) = J_w$.
Moreover, the transformation $S = \chi(T)$ is the transformation induced by $T$ on $J_w$.
By Theorem~\ref{theo:inductionbi} there is an automorphism $\theta$ of the free group on $A$ such that $\Sigma_T(z) = \theta(\Sigma_S(z))$ for any $z \in J_w$.

By Lemma~\ref{lem:returnsind}, we have $x \in \RR_F(w)$ if and only if $\Sigma_T(z) = x\Sigma_T(S(z)))$ for some $z \in J_w$.
This implies that $\RR_F(w) = \theta(A)$.
Indeed, for any $z \in J_w$, let $a$ is the first letter of $\Sigma_S(z)$.
Then
$$
\Sigma_T(z) = \theta(\Sigma_S(z)) = \theta(a\Sigma_S(S(z))) = \theta(a)\theta(\Sigma_S(Sz)) = \theta(a)\Sigma_T(S(z)).
$$
Thus $x \in \RR_F(w)$ if and only if there is $a\in A$ such that $x = \theta(a)$.
This proves that the set $\RR_F(w)$ is a basis of the free group on $A$.
\end{proof}

The property proved in the previous corollary is actually true for a much larger class of sets than regular interval exchange sets (see~\cite[Theorem 4.7]{BertheDeFeliceLeroyPerrinReutenauerRindone2014}).
We illustrate the this result with the following examples.

\begin{example}
We consider again the transformation $T$ of Example~\ref{ex:2alpha} and $F = F(T)$.
We have $R_F(c) = \{bac,bbac,c\}$ (see Example~\ref{ex:returns}).
We represent in Figure~\ref{fig:illustrate} the sequence $\chi$ of Rauzy inductions such that $J_c$ is the domain of $\chi(T)$.

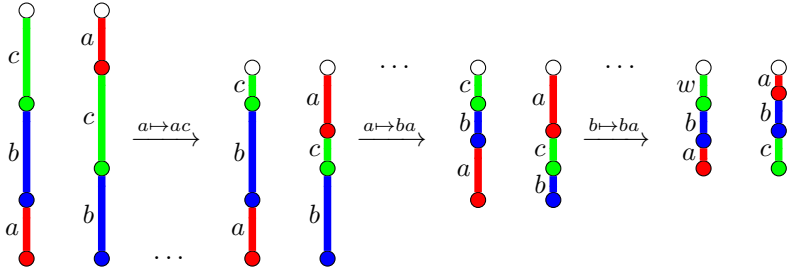
\begin{figure}[hbt]
\centering
\gasset{Nadjust=wh}
\begin{picture}(100,40)

\put(0,0){
\begin{picture}(20,20)
\node[fillcolor=red](0h)(0,0){}
\node[fillcolor=blue](1-2alpha)(0,7.8){}
\node[fillcolor=green](1-alpha)(0,20.6){}
\node(1h)(0,33){}
\drawedge[linecolor=red,linewidth=1](0h,1-2alpha){$a$}
\drawedge[linecolor=blue,linewidth=1](1-2alpha,1-alpha){$b$}
\drawedge[linecolor=green,linewidth=1](1-alpha,1h){$c$}

\node[fillcolor=blue](0b)(10,0){}
\node[fillcolor=green](alpha)(10,12){}
\node[fillcolor=red](2alpha)(10,25.4){}\node(1b)(10,33){}
\drawedge[linecolor=blue,linewidth=1](0b,alpha){$b$}
\drawedge[linecolor=green,linewidth=1](alpha,2alpha){$c$}
\drawedge[linecolor=red,linewidth=1](2alpha,1b){$a$}
\end{picture}
}
\put(15,15){$\edge{a\mapsto ac}$}\put(18,0){$\ldots$}
\put(30,0){
\begin{picture}(20,20)
\node[fillcolor=red](0h)(0,0){}
\node[fillcolor=blue](1-2alpha)(0,7.8){}
\node[fillcolor=green](1-alpha)(0,20.6){}
\node(2alphah)(0,25.4){}
\drawedge[linecolor=red,linewidth=1](0h,1-2alpha){$a$}
\drawedge[linecolor=blue,linewidth=1](1-2alpha,1-alpha){$b$}
\drawedge[linecolor=green,linewidth=1](1-alpha,2alphah){$c$}

\node[fillcolor=blue](0b)(10,0){}
\node[fillcolor=green](alpha)(10,12){}
\node[fillcolor=red](4alpha-1)(10,17){}
\node(2alpha)(10,25.4){}
\drawedge[linecolor=blue,linewidth=1](0b,alpha){$b$}
\drawedge[linecolor=green,linewidth=1](alpha,4alpha-1){$c$}
\drawedge[linecolor=red,linewidth=1](4alpha-1,2alpha){$a$}
\end{picture}
}
\put(45,15){$\edge{a\mapsto ba}$}\put(48,25.4){$\ldots$}
\put(60,0){
\begin{picture}(20,20)
\node[fillcolor=red](0h)(0,7.8){}
\node[fillcolor=blue](1-2alpha)(0,15.7){}
\node[fillcolor=green](1-alpha)(0,20.6){}
\node(2alpha)(0,25.4){}
\drawedge[linecolor=red,linewidth=1](0h,1-2alpha){$a$}
\drawedge[linecolor=blue,linewidth=1](1-2alpha,1-alpha){$b$}
\drawedge[linecolor=green,linewidth=1](1-alpha,2alpha){$c$}

\node[fillcolor=blue](0b)(10,7.8){}
\node[fillcolor=green](alpha)(10,12){}
\node[fillcolor=red](4alpha-1)(10,17){}
\node(2alpha)(10,25.4){}
\drawedge[linecolor=blue,linewidth=1](0b,alpha){$b$}
\drawedge[linecolor=green,linewidth=1](alpha,4alpha-1){$c$}
\drawedge[linecolor=red,linewidth=1](4alpha-1,2alpha){$a$}
\end{picture}
}
\put(75,15){$\edge{b\mapsto ba}$}\put(78,25.4){$\ldots$}
\put(90,0){
\begin{picture}(20,20)
\node[fillcolor=red](0h)(0,12){}
\node[fillcolor=blue](1-2alpha)(0,15.7){}
\node[fillcolor=green](1-alpha)(0,20.6){}
\node(2alpha)(0,25.4){}
\drawedge[linecolor=red,linewidth=1](0h,1-2alpha){$a$}
\drawedge[linecolor=blue,linewidth=1](1-2alpha,1-alpha){$b$}
\drawedge[linecolor=green,linewidth=1](1-alpha,2alpha){$w$}

\node[fillcolor=green](alpha)(10,12){}
\node[fillcolor=blue](4alpha-1)(10,17){}
\node[fillcolor=red](7alpha-2)(10,22){}
\node(2alpha)(10,25.4){}
\drawedge[linecolor=green,linewidth=1](alpha,4alpha-1){$c$}
\drawedge[linecolor=blue,linewidth=1](4alpha-1,7alpha-2){$b$}
\drawedge[linecolor=red,linewidth=1](7alpha-2,2alpha){$a$}
\end{picture}
}
\end{picture}
\caption{The sequence $\chi \in \{\varphi,\psi\}^*$}
\label{fig:illustrate}
\end{figure}

The sequence is composed of a right induction followed by two left inductions.
We have indicated on each edge the associated automorphism (indicating only the image of the letter which is modified).
We have $\chi = \varphi^2\psi$ and the resulting composition $\theta$ of automorphisms gives
$$
\theta(a) = bac, \quad \theta(b) = bbac, \quad \theta(c)=c.
$$
Thus $\RR_F(c) = \theta(A)$.
\end{example}

\begin{example}
\label{ex:u}
Let $T$ and $F$ be as in the preceding example.
Let $U$ be the transformation induced by $T$ on $J_a$.
We have $U = \varphi^6(T)$ and a computation shows that for any $z \in J_a$, $\Sigma_T(z) = \theta(\Sigma_U(z))$ where $\theta$ is the automorphism of the free group on $A = \{a,b,c\}$ which is the coding morphism for $\RR_F(a)$ defined by:
$$
\theta(a) = ccba, \quad \theta(b) = cbba, \quad \theta(c)=ccbba.
$$
One can verify that $F(U) = F(S)$, where $S$ is the transformation obtain from $T$ by permuting the labels of the intervals according to the permutation $\pi = (acb)$.

Note that $F(U) = F(S)$ although $S$ and $U$ are not identical, even up to rescaling the intervals.
Actually, the rescaling of $U$ to a transformation on $[0,1[$ corresponds to the mirror image of $S$, obtained by taking the image of the intervals by a symmetry centered at $1/2$.
\end{example}

Note that in the above examples, all lengths of the intervals belong to the quadratic number field $\Q[\sqrt{5}]$.

In the next Section we will prove that if a regular interval exchange transformation $T$ is defined over a quadratic field, then the family of transformations obtained from $T$ by the Rauzy inductions contains finitely many distinct transformations up to rescaling.

\section{Interval exchange over a quadratic field}
\label{sec:quadratic}
An interval exchange transformation is said to be defined over a set $Q \subset \R$ if the lengths of all exchanged semi-intervals belong to $Q$.

The following is proved in~\cite{BoshernitzanCarroll1997}.
Let $T$ be a minimal interval exchange transformation on semi-intervals defined over a quadratic number field.
Let $(T_n)_{n\ge 0}$ be a sequence of interval exchange transformation such that $T_0=T$ and $T_{n+1}$ is the transformation induced by $T_n$ on one of its exchanged semi-intervals $I_n$. Then, up to rescaling all semi-intervals $I_n$ to the same length, the sequence $(T_n)$ contains finitely many distinct transformations.
In the same paper, an extension to the right Rauzy induction is suggested (but not completly developed).

In this section we generalize this results and prove that, under the above hypothesis on the lengths of the semi-intervals and up to rescaling and translation, there are finitely many transformations obtained by the branching Rauzy induction defined in Section~\ref{sec:rauzy}.

\begin{theorem}
\label{theo:quadratic}
Let $T$ be a regular interval exchange transformation defined over a quadratic field.
The family of all induced transformation of $T$ over an admissible semi-interval contains finitely many distinct transformations up to equivalence.
\end{theorem}

The proof of the Theorem~\ref{theo:quadratic} is based on the fact that for each minimal interval exchange transformation defined over a quadratic field, a certain measure of the arithmetic complexity of the admissible semi-intervals is bounded.

\subsection{Complexities}
\label{subsec:complexities}

Let $T$ be an interval exchange transformation on a semi-interval $[\ell, r[$ defined over a quadratic field $ \Q[\sqrt{d}]$, where $d$ is a square free integer $\geq 2$.
Without loss of generality, one may assume, by replacing $T$ by an equivalent interval exchange transformation if necessary, that $T$ is defined over the ring $ \Z[\sqrt{d}] = \{ m + n\sqrt{d} \,\, | \,\, m,n \in \Z \}$ and that all $\gamma_i$ and $\alpha_i$ lie in $\Z[\sqrt{d}]$ (replacing $[\ell, r[$ if necessary by its equivalent translate with $\gamma_0 = \ell \in \Z[\sqrt{d}]$).

For $z = m + n \sqrt{d}$, let define $\Psi(z) = \max ( |m|, |n| )$.

Let $\mathcal{A}([\ell,r[)$ be the algebra of subsets $X \subset [\ell,r[$ which are finite unions $X = \bigcup_j I_j$ of semi-intervals defined over $\Z[\sqrt{d}]$, that is $I_j = [\ell_j, r_j[$ for some $\ell_j, r_j \in \Z[\sqrt{d}]$.
Note that the algebra $\mathcal{A}([\ell,r[)$ is closed under taking finite unions, intersections and passing to complements in $[\ell, r[$.

Set $\partial(X)$ the boundary of $X$ and $|X|$ the Lebesgue measure of $X$.
Given a subset $X \in \mathcal{A}([\ell,r[)$, we define the \emph{complexity} of $X$ as $\Psi(X) = \max \{ \Psi(z) \,\, | \,\, z \in \partial(X) \}$ and the \emph{reduced complexity} of $X$ as $\Pi(X) = |X| \, \Psi(X)$.

A key tool to prove Theorem~\ref{theo:quadratic} is the following result proved in~\cite[Theorem 3.1]{BoshernitzanCarroll1997}.

\begin{theorem}[Boshernitzan]
\label{theo:bosh}
Let $T$ be a minimal interval exchange transformation on an interval $[\ell,r[$ defined over a quadratic number field. Assume that $(J_n)_{\geq 1}$ is a sequence of semi-intervals of $[\ell,r[$ such that the set $\{ \Pi(J_n) \,\, | \,\, n \geq 1 \}$ is bounded.
Then the sequence $T_n$ of interval exchange transformations obtained by inducing $T$ on $J_n$ contains finitely many distinct equivalence classes of interval exchange transformations.
\end{theorem}

Thus, in order to prove Theorem~\ref{theo:quadratic}, it is sufficient to show that the reduced complexity of every admissible semi-interval is bounded.

The following Proposition is proved in~\cite[Proposition 2.1]{BoshernitzanCarroll1997}.
It shows that the complexity of a subset $X$ and of its image $T(X)$ differ at most by a constant that depends only on $T$.

\begin{proposition}
\label{pro:psiT-psi}
There exists a constant $u = u(T)$ such that for every $X \in \mathcal{A}([\ell,r[)$ and $z \in [\ell,r[$ one has $| \Psi(T(X)) - \Psi(X) | \leq u$ and $\Psi(T(z) - z) \leq u.$
Moreover, one has $\Psi(\gamma) \leq u$ and $\Psi(T(\gamma)) \leq u$ for every separation point $\gamma$.
\end{proposition}

Clearly, by Proposition~\ref{pro:psiT-psi}, one also has $| \Psi(T^{-1}(X)) - \Psi(X) | \leq u$ for every $X \in \mathcal{A}([\ell,r[)$ and $\Psi(T^{-1}(z) - z) \leq u$ for every $z \in [\ell,r[$.

Although it is not necessary for our purposes, we can improve the approximation of the reduced complexity of a nonempty subset $X \in \mathcal{A}([\ell, r|)$ by the following proposition. This result, proved in~\cite[Proposition 2.4]{BoshernitzanCarroll1997}, determines a lower bound on $\Pi(X)$.

\begin{proposition}
\label{pro:piS}
Let $X \in \mathcal{A}([\ell,r[)$ be a subset composed of $n$ disjoints semi-intervals. Then $\Pi(X) > n / (4 \sqrt{d})$.
\end{proposition}

\subsection{Return times}
\label{subsec:return}

Let $T$ be an interval exchange transformation.
For a subset $X \in \mathcal{A}([\ell, r[)$ we define the \emph{maximal positive return time} and \emph{maximal negative return time} for $T$ on $X$ by the functions
$$
\textstyle{\rho^+(X) = \min \left\{ n \geq 1 \, | \, T^n(X) \subset \bigcup_{i = 0}^{n-1} T^i(X) \right\},}
$$
and
$$
\textstyle{\rho^-(X) = \min \left\{ m \geq 1 \, | \, T^m(X) \subset \bigcup_{i = 0}^{m-1} T^{-i}(X) \right\}.}
$$

We also define the \emph{minimal positive return time} and the \emph{minimal negative return time} as
$$
\sigma^+(X) = \min \left\{ n \geq 1 \, | \, T^n(X) \cap X \neq \emptyset \right\},
$$
and
$$
\sigma^-(X) = \min \left\{ m \geq 1 \, | \, T^{-m}(X) \cap X \neq \emptyset \right\}.
$$
If $T$ is minimal, it is clear that for every $X \in \mathcal{A}([\ell, r[)$, one has
$$
\textstyle{[\ell, r[ \, = \bigcup_{i=0}^{\rho^+(X)-1} T^i(X) = \bigcup_{i=0}^{\rho^-(X)-1} T^{-i}(X).}
$$

Note that when $J$ is a semi-interval, we have $\rho^+(J) = \max_{z \in J} \rho_{J,T}^+(z)$ and $\sigma^+(J) = \min_{z \in J} \rho_{J,T}^+(z)$.
Symmetrically $\rho^-(J) = \max_{z \in J} \rho_{J,T}^-(z) + 1$ and $\sigma^-(J) = \min_{z \in J} \rho_{J,T}^-(z) + 1$.

Let $\zeta, \eta$ be two functions. We write $\zeta \in O(\eta)$ if there exists a constant $C$ such that $|\zeta| \leq C |\eta|$. We write $\zeta \in \Theta(\eta)$ if one has both $\zeta \in O(\eta)$ and $\eta \in O(\zeta)$.
Note that $\Theta$ is an equivalence relation, that is $\zeta \in \Theta(\eta) \Leftrightarrow \eta \in \Theta(\zeta)$.

Boshernitzan and Carroll give in~\cite{BoshernitzanCarroll1997} two upper bounds for $\rho^+(X)$ and $\sigma^+(X)$ for a subset $X$ (Theorems 2.5 and 2.6 respectively) and a more precise estimation when the subset is a semi-interval (Theorem 2.8).
Some slight modifications of the proofs can be made so that the results hold also for $\rho^-$ and $\sigma^-$.
We summarize these estimates in the following theorem.

\begin{theorem}
\label{theo:sim}
For every $X \in \mathcal{A}([\ell,r[)$ one has
$ \rho^+(X), \rho^-(X) \in O(\Psi(X))$
and
$\sigma^+(X), \sigma^-(X) \in O\left( 1 / |X| \right)$.
Moreover, if $T$ is minimal and $J$ is a semi-interval, then
$\textstyle{\rho^+(J) \in \Theta\left(\rho^-(J)\right) = \Theta\left(\sigma^+(J)\right) = \Theta\left(\sigma^-(J)\right) = \Theta\left( 1 /|J| \right)}$.
\end{theorem}

An immediate corollary of Theorem~\ref{theo:sim} is the following
\begin{corollary}
\label{cor:oM}
Let $T$ be minimal and assume that
$$
\{ T^i(z) \, | \, -m+1 \leq i \leq n-1 \} \cap J = \emptyset
$$
for some point $z \in [\ell,r[$, some semi-interval $J \subset [\ell,r[$ and some integers $m, n~\geq~1$.
Then $|J| \in O\left( 1 / \max \{ m,n \} \right)$.
\end{corollary}
\begin{proof}
By the hypothesis, $z \notin \bigcup_{i=0}^{n-1} T^{-i}(J)$, then we have $\rho^-(J) \geq n$.
By Theorem~\ref{theo:sim}, we obtain
$ |J| \in \Theta\left( 1 / \rho^-(J) \right) \subseteq O\left( 1 / n \right)$.
Symmetrically, since $\rho^+(J) \geq m$, one has $ |J| \in O\left( 1 / m \right)$.
Then
$\textstyle{|J| \in O\left( \min \left\{ 1 / m, 1 / n \right\} \right) = O\left( 1 / \max \{ m,n \} \right)}$.
\end{proof}

\subsection{Reduced complexity of admissible semi-intervals}
\label{subsec:main}

In order to obtain Theorem~\ref{theo:quadratic}, we prove some preliminary results concerning the reduced complexity of admissible semi-intervals.

Let $T$ be an $s$-interval exchange transformation.
Recall from Section~\ref{sec:ie} that we denote by $\Sep(T) = \{ \gamma_i \, | \, 0 \leq i \leq s-1 \}$ the set of separation points.
For every $n \geq 0$ define
$\mathcal{S}_n(T) = \bigcup_{i=0}^{n-1} T^{-i} \big(\Sep(T)\big)$
with the convention $\mathcal{S}_0 = \emptyset$.

Since $\Sep(T^{-1}) = T\big(\Sep(T)\big)$, one has $\mathcal{S}_n(T^{-1}) = T^{n-1}\big(\mathcal{S}_n(T)\big)$.

Given two integers $m,n \geq 1$, we can define $\mathcal{S}_{m,n} = \mathcal{S}_m(T) \cup \mathcal{S}_n(T^{-1})$.
An easy calculation shows that $\mathcal{S}_{m,n}(T) = \bigcup_{i=-m+1}^{n} T^i\big(\Sep(T)\big).$
Observe also that $\mathcal{S}_{m,n}(T) = T^{n}\big( \mathcal{S}_{m+n}(T)\big) = T^{-m+1}\big( \mathcal{S}_{m+n}(T)\big)$.

Denote by $\mathcal{V}_{m,n}(T)$ the family of semi-intervals whose endpoints are in $\mathcal{S}_{m,n}(T)$.
Put $\mathcal{V}(T) = \bigcup_{m,n \geq 0} \mathcal{V}_{m,n}(T)$.
Every admissible semi-interval belongs to $\mathcal{V}(T)$, while the converse is not true.

\begin{theorem}
\label{theo:pi}
$\Pi(J) \in \Theta(1)$ for every semi-interval $J$ admissible for $T$.
\end{theorem}
\begin{proof}
Let $m, n$ be the two minimal integers such that $J = [t,w[ \, \in \mathcal{V}_{m,n}(T)$.
Then $t,w \in \{T^m(\gamma_i) \, | \, 1 \leq i \leq s \} \cup \{T^{-n}(\gamma_i) \, | \, 1 \leq i \leq s \}$.
Suppose, for instance, $t = T^M(\gamma)$, with $M = \max \{m,n\}$ and $\gamma$ a separation point.
The other cases (namely, $t = T^{-M}(\gamma)$, $w=T^{M}(\gamma)$ or $w=T^{-M}(\gamma)$) are proved similarly.

The only semi-interval in $\mathcal{V}_{0,0}(T)$ is $[\ell,r[$ and clearly in this case the theorem is verified.

Suppose then that $J \in \mathcal{V}_{m,n}(T)$ for some nonnegative integers $m,n$ with $m+n > 0$.
We have
$\Psi(J) = \max \{ \Psi(t), \Psi(w) \} \leq Mu$ where $u$ is the constant introduced in Proposition~\ref{pro:psiT-psi}.
Moreover, by the definition of admissibility one has
$ \{ T^j(\gamma) \, | \, 1 \leq j \leq M \} \cap J = \emptyset$.
Thus, by Corollary~\ref{cor:oM} we have $|J| \in O( 1 /M )$.
Then
$ \Pi(J) = |J| \ \Psi(J) \in O(1)$.
By Proposition~\ref{pro:piS} we have $\Pi(J) > 1 / (4 \sqrt{d})$.
This concludes the proof.
\end{proof}

Denote by $\mathcal{U}_{m,n}(T)$ the family of semi-intervals partitioned by $\mathcal{S}_{m,n}(T)$.
Clearly $\mathcal{V}_{m,n}(T)$ contains $\mathcal{U}_{m,n}(T)$.
Indeed every semi-interval $J \in \mathcal{V}_{m,n}(T)$ is a finite union of contiguous semi-intervals belonging to $\mathcal{U}_{m,n}(T)$.

Note that $\mathcal{U}_{m,0}(T)$ is the family of semi-intervals exchanged by $T^m$, while $\mathcal{U}_{0,n}(T)$ is the family of semi-intervals exchanged by $T^{-n}$.

Put $\mathcal{U}(T) = \bigcup_{m,n \geq 0} \mathcal{U}_{m,n}(T)$.
Using Theorem~\ref{theo:pi} we easily deduce the following corollary, which is a generalization of Theorem 2.11 in \cite{BoshernitzanCarroll1997}.

\begin{corollary}
\label{cor:pi}
$\Pi(J) \in \Theta(1)$ for every semi-interval $J \in \mathcal{U}(T)$.
\end{corollary}

We are now able to prove Theorem~\ref{theo:quadratic}.

\begin{proofof}{ of Theorem~\ref{theo:quadratic}}
By Theorem~\ref{theo:birauzy2}, every admissible semi-interval can be obtained by a finite sequence $\chi$ of right and left Rauzy inductions.
Thus we can enumerate the family of all admissible semi-intervals.
The conclusion easily follows from Theorem~\ref{theo:bosh} and Theorem~\ref{theo:pi}.
\end{proofof}

An immediate corollary of Theorem~\ref{theo:quadratic} is the following.

\begin{corollary}
\label{cor:graph}
Let $T$ be a regular interval exchange transformation defined over a quadratic field. Then the induction graph $\mathcal{G}(T)$ and the modified induction graph $\widetilde{\mathcal{G}}(T)$ are finite.
\end{corollary}

\begin{example}
Let $T$ be the regular interval exchange transformation of Example~\ref{ex:2alpha}.
The modified induction graph $\widetilde{\mathcal{G}}(T)$ is represented in Figure~\ref{fig:graph}.
The transformation $T$ belongs to the similarity class $\langle T_1 \rangle$ as well as transformations $S$ of Example~\ref{ex:simphi6T} and $U$ of Example~\ref{ex:phi6T}.
The transformations $\psi(T)$ and $\psi^2(T)$ of Example~\ref{ex:induced2} belongs respectively to classes $\langle T_2 \rangle$ and $\langle T_4 \rangle$, while the two last transformations of Figure~\ref{fig:illustrate}, namely $\varphi \psi(T)$ and $\varphi^2 \psi(T)$, belongs respectively to $\langle T_5 \rangle$ and $\langle T_7 \rangle$.
Finally, the left Rauzy induction sequence from $T$ to $U = \varphi^6 (T)$ corresponds to the loop $\langle T_1 \rangle \to \langle T_3 \rangle \to \langle T_4 \rangle \to \langle T_6 \rangle \to \langle T_7 \rangle \to \langle T_8 \rangle \to \langle T_1 \rangle$ in $\widetilde{\mathcal{G}}(T)$.

\begin{figure}[hbt]
\centering
\gasset{Nframe=y}
\begin{picture}(100,70)(-50,-35)

\node(1)(-50,0){$\langle T_1 \rangle$}
\node(2)(-10,15){$\langle T_3 \rangle$}
\node(3)(-10,-15){$\langle T_2 \rangle$}
\node(4)(20,0){$\langle T_4 \rangle$}
\node(5)(5,0){$\langle T_5 \rangle$}
\node(6)(20,15){}
\node(7)(20,-15){$\langle T_6 \rangle$}
\node(8)(-10,0){$\langle T_7 \rangle$}
\node(9)(-30,0){}
\node(10)(-30,15){}
\node(11)(-30,-15){$\langle T_8 \rangle$}
\node(13)(5,25){}
\node(14)(5,-25){}
\node(15)(50,0){}
\node(16)(35,15){}
\node(17)(35,-15){}
\node(18)(35,0){}

\drawedge[curvedepth=15,linewidth=0.4](1,2){}
\drawedge[curvedepth=-15](1,3){}
\drawedge[curvedepth=1,linewidth=0.4](2,4){}
\drawedge[curvedepth=0](2,5){}
\drawedge[curvedepth=-1](3,4){}
\drawedge[curvedepth=0](3,5){}
\drawedge[curvedepth=0](4,6){}
\drawedge[curvedepth=0,linewidth=0.4](4,7){}
\drawedge[curvedepth=0](5,8){}
\drawedge[curvedepth=-3](6,8){}
\drawedge[curvedepth=-3](6,9){}
\drawedge[curvedepth=3,linewidth=0.4](7,8){}
\drawedge[curvedepth=3](7,9){}
\drawedge[curvedepth=0](8,10){}
\drawedge[curvedepth=0,linewidth=0.4](8,11){}
\drawedge[curvedepth=0](9,10){}
\drawedge[curvedepth=0](9,11){}
\drawedge[curvedepth=0](10,1){}
\drawedge[curvedepth=10](10,13){}
\drawedge[curvedepth=0,linewidth=0.4](11,1){}
\drawedge[curvedepth=-10](11,14){}
\drawedge[curvedepth=0](13,2){}
\drawedge[curvedepth=15](13,15){}
\drawedge[curvedepth=0](14,3){}
\drawedge[curvedepth=-15](14,15){}
\drawedge[curvedepth=0](15,16){}
\drawedge[curvedepth=0](15,17){}
\drawedge[curvedepth=0](16,4){}
\drawedge[curvedepth=0](16,18){}
\drawedge[curvedepth=0](17,4){}
\drawedge[curvedepth=0](17,18){}
\drawedge[curvedepth=0](18,6){}
\drawedge[curvedepth=0](18,7){}

\end{picture}
\caption{Modified induction graph of the transformation $T$.}
\label{fig:graph}
\end{figure}
\end{example}

\subsection{Primitive morphic sets}
\label{subsec:morphic}

In this section we show an important property of interval exchange transformations defined over a quadratic field, namely that the related interval exchange sets are primitive morphic.
Let prove first the following result.

\begin{proposition}
\label{pro:primitive}
Let $T, \chi(T)$ be two equivalent regular interval exchange transformations with $\chi \in \left\{ \varphi, \psi \right\}^*$.
There exists a primitive morphism $\theta$ and a point $z \in D(T)$ such that the natural coding of $T$ relative to $z$ is a fixed point of $\theta$.
\end{proposition}
\begin{proof}
By Proposition~\ref{pro:regularur}, the set $F(T)$ is uniformly recurrent.
Thus, there exists a positive integer $N$ such that every letter of the alphabet appears in every word of length $N$ of $F(T)$.
Moreover, by Theorem~\ref{theo:length}, applying iteratively the Rauzy induction, the length of the domains tends to zero.

Consider $T' = \chi^m (T)$, for a positive integer $m$, such that $D(T')~<~\varepsilon$, where $\varepsilon$ is the positive real number for which, by Lemma~\ref{lem:distance}, the first return map for every point of the domain is ``longer'' than $N$, i.e. $T'(z) = T^{n(z)}(z)$, with $n(z) \geq N$, for every $z \in D(T')$.

By Theorem~\ref{theo:inductionbi} and the remark following it, there exists an automorphism $\theta$ of the free group and a point $z \in D(T') \subseteq D(T)$ such that the natural coding of $T$ relative to $z$ is a fixed point of $\theta$, that is $\Sigma_T (z) = \theta\left( \Sigma_{T} (z) \right)$.

By the previous argument, the image of every letter by $\theta$ is longer than $N$, hence it contains every letter of the alphabet as a factor.
Therefore, $\theta$ is a primitive morphism.
\end{proof}

\begin{theorem}
\label{theo:morphic}
Let $T$ be a regular interval exchange transformation defined over a quadratic field.
The interval exchange set $F(T)$ is primitive morphic.
\end{theorem}
\begin{proof}
By Theorem~\ref{theo:quadratic} there exists a regular interval transformation $S$ such that we can find in the induction graph $\mathcal{G}(T)$ a path from $[T]$ to $[S]$ followed by a cycle on $[S]$.
Thus, by Theorem~\ref{theo:inductionbi} there exists a point $z \in D(S)$ and two automorphisms $\theta, \eta$ of the free group such that $\Sigma_T(z) = \theta \left( \Sigma_S (z) \right)$, with $\Sigma_S (z)$ a fixed point of $\eta$.

By Proposition~\ref{pro:primitive} we can suppose, without loss of generality, that $\eta$ is primitive.
Therefore, $F(T)$ is a primitive morphic set.
\end{proof}

\begin{example}
Let $T = T_{\lambda, \pi}$ be the transformation of Example~\ref{ex:2alpha} (see also~\ref{ex:t}).
The set $F(T)$ is primitive morphic.
Indeed the transformation $T$ is regular and the length vector $\lambda = (1-2\alpha, \alpha, \alpha)$ belongs to $\Q\left[\sqrt{5}\right]^3$.
\end{example}

\bibliographystyle{plain}
\bibliography{branching}

\end{document}